\documentclass[11pt,onecolumn]{IEEEtran}

\makeatletter
\long\def\@makecaption#1#2{\ifx\@captype\@IEEEtablestring%
\footnotesize\begin{center}{\normalfont\footnotesize #1}\\
{\normalfont\footnotesize\scshape #2}\end{center}%
\@IEEEtablecaptionsepspace \else \@IEEEfigurecaptionsepspace
\setbox\@tempboxa\hbox{\normalfont\footnotesize {#1.}~~ #2}%
\ifdim \wd\@tempboxa >\hsize%
\setbox\@tempboxa\hbox{\normalfont\footnotesize {#1.}~~ }%
\parbox[t]{\hsize}{\normalfont\footnotesize \noindent\unhbox\@tempboxa#2}%
\else \hbox
to\hsize{\normalfont\footnotesize\hfil\box\@tempboxa\hfil}\fi\fi}
\makeatother

\usepackage{cite}
\usepackage{graphicx}
\graphicspath{{figs/}}
\usepackage[font=scriptsize]{subfig}
\usepackage{amssymb}
\usepackage{amsmath}
\usepackage{comment}
\usepackage[dvipsnames]{xcolor}
\usepackage{algorithm}
\usepackage[noend]{algorithmic}
\usepackage{url}
\usepackage[T1]{fontenc}

\newcommand{\eg}{\textit{e.g.},}
\newcommand{\ie}{\textit{i.e.},}

\newcommand{\eqend}{\,\cdot}

\newcommand{\pr}[1]{\mathbb{P}\left\{ #1 \right\}}
\newcommand{\E}[1]{\mathbb{E}\left[#1\right]}
\newcommand{\mx}[1]{\ensuremath{\mathbf{#1}}}
\newcommand{\mc}[1]{\mathcal{#1}}
\newcommand{\pmax}{ \ensuremath{P_{\mathrm{max}}} }

\newcommand{\set}[1]{\left\{ #1 \right\}}

\newcommand{\enc}[1]{ (#1) }
\newcommand{\angles}[1]{ \langle#1\rangle }
\newcommand{\sinr}{\mathrm{SINR}}
\newcommand{\fnote}[1]{\footnote{#1}}
\newcommand{\mg}[1]{}
\newcommand{\spar}[1]{\smallskip\noindent{\bf#1. }}

\newtheorem{lemma}{Lemma}
\newtheorem{theorem}{Theorem}
\newtheorem{propose}{Proposition}

\begin{document}

\title{Efficient Wireless Security\\ Through Jamming, Coding and Routing}
\author{
    Majid Ghaderi$^*$,
    \thanks{$^*$Department of Computer Science, University of Calgary, Email: {\tt mghaderi@ucalgary.ca}}
    Dennis Goeckel$^{\dag}$,
    \thanks{$^{\dag}$Department of Electrical and Computer Engineering,
    UMass Amherst, \mbox{Emails: {\tt \{goeckel, mdehghan\}@ece.umass.edu}  }}
    Ariel Orda$^{\ddag}$ and
    \thanks{$^{\ddag}$Department of Electrical Engineering, Technion, Email: {\tt ariel@ee.technion.ac.il}}
    Mostafa Dehghan$^{\dag}$
}

\maketitle

\begin{abstract}
There is a rich recent literature on how to assist secure communication
between a single transmitter and receiver at the physical layer of
wireless networks through techniques such as cooperative jamming.
In this paper, we consider how these single-hop physical layer
security techniques can be extended to multi-hop wireless networks
and show how to augment physical layer security techniques with higher
layer network mechanisms such as coding and routing. Specifically, we
consider the secure minimum energy routing problem, in which the
objective is to compute a minimum energy path between two network
nodes subject to constraints on the end-to-end communication
secrecy and goodput over the path.  This problem is formulated
as a constrained optimization of transmission power and link
selection, which is proved to be \mbox{NP-hard}. Nevertheless,
we show that efficient algorithms exist to compute both exact and
approximate solutions for the problem. In particular, we develop an
exact solution of pseudo-polynomial complexity, as well as an $\epsilon$-optimal
approximation of polynomial complexity. Simulation results are also
provided to show the utility of our algorithms and quantify their
energy savings compared to a combination of (standard) security-agnostic
minimum energy routing and physical layer security.  In the
simulated scenarios, we observe that, by jointly optimizing
link selection at the network layer and cooperative jamming at the
physical layer, our algorithms reduce the network energy consumption
by half.
\end{abstract}


\IEEEpeerreviewmaketitle

\section{Introduction}
\label{sec:intro}

Protecting the secrecy of user messages has become a major concern
in modern communication networks.  Due to the propagation properties
of the wireless medium, wireless networks can potentially make the
problem more challenging by allowing an eavesdropper to have
relatively easy access to the transmitted message if countermeasures
are not employed.  Our goal is to provide everlasting
security in this wireless environment; that is, we will consider
methods that will prevent an eavesdropper from ever decoding a
transmitted message - even if the eavesdropper has the capability
to record the signal and attempt decryption over many years (or decades).
There are two different classes of security techniques of interest
here:  cryptographic approaches based on computational complexity,
and information-theoretic approaches that attempt to obtain
perfect secrecy.  Both have advantages and disadvantages
for the desired everlasting security in the wireless environment.

The traditional solution to providing security in a wireless
environment is the cryptographic approach: assume that the
eavesdropper will get the transmitted signal without distortion,
but the desired recipient who shares a key with the transmitter
is able to decode the message easily, while the eavesdropper
lacking the key must solve a \emph{hard} problem that is beyond
her/his computational capabilities~\cite{stinson2006cryptography}.
Since the eavesdropper is assumed to get the transmitted signal
without distortion, cryptography addresses the key challenge in the
wireless environment of thwarting an eavesdropper very near the
transmitter.  However, such an approach faces the concern that the
eavesdropper can store the signal, and, then, with later advances in
computational capabilities or by breaking the encryption scheme, obtain
the message.  The desire for everlasting security then motivates
adding countermeasures at the physical layer that inhibit even the
recording of the encrypted message by the eavesdropper that combine
with cryptography to facilitate a defense-in-depth approach~\cite{NSA_defense}.

In the information-theoretic approach to obtain perfect
secrecy~\cite{shannon1949communication}, the goal is to guarantee
that the eavesdroppers can never extract information from the
message, regardless of their computational capability.
Wyner~\cite{wyner1975wire} and succeeding
authors~\cite{leung1978gaussian,csiszar1978broadcast} showed that
perfect secrecy is possible if the channel conditions between the
transmitter and receiver were favorable relative to the channel
conditions between the transmitter and eavesdropper. In this
so-called \emph{wiretap} channel, perfect secrecy at a positive rate
with no pre-shared key is possible~\cite{wyner1975wire}.  This
clearly satisfies the requirement for everlasting secrecy, but it
relies on favorable channel conditions that are difficult (if not
impossible) to guarantee in a wireless environment.  Hence,
information-theoretic secrecy requires a network design which
inhibits reception at the eavesdropper while supporting reception at
the desired recipient.

Our work supports both a cryptographic (computational) approach or
information-theoretic approach.  Per above, it is advantageous in
either case to seek or create conditions so as to inconvenience
reception at eavesdropper(s) while facilitating communication of the
legitimate system nodes. This has been actively considered in the
literature on the physical layer of wireless networks over the last
decade, with approaches based on both
opportunism~\cite{maurer93,bloch2008} and active channel
manipulation~\cite{goel08,tekin2008general} being employed. Most of
these works have arisen in the information-theoretic community and
considered small networks consisting of a source, destination,
eavesdropper, and perhaps a relay
node(s)~\cite{bloch2008,goel08,tekin2008general,dong10,elgamal08,
Dennis2011JSAC}. More recently, there has been the active
consideration of large networks with the introduction of the secrecy
graph to consider secure
connectivity~\cite{haenggi2008secrecy,pinto2008physical,pinto2009wireless}
and a number of approaches to throughput scaling versus security
tradeoffs~\cite{liang2009secrecy, koyluoglu2010secrecy,
vasudevan2010security}.  Hence, whereas there has been a significant
consideration of small single- and two-hop networks and
asymptotically large multi-hop networks, there has been almost no
consideration of the practical multi-hop networks that lie between
those two extremes. It is this large and important gap that this
paper fills.

Consider a network where system nodes communicate with each other
wirelessly, possibly over multiple hops, such as in wireless mesh
networks and ad hoc networks. A set of {\em eavesdroppers} try to passively listen
to communications among legitimate network nodes. To prevent the
eavesdroppers from successfully capturing communications between
legitimate nodes, mechanisms to thwart such are employed at the
physical layer of the network. Two nodes that wish
to communicate securely may need to do so over multiple hops in
order to thwart eavesdroppers or simply because the nodes are
not within the reach of each other. While we make no argument about
the optimality or practicality of any specific physical layer
security mechanism, for the sake of concreteness, we focus on
\emph{cooperative jamming}, which has received considerable attention
~\cite{goel08,tekin2008general,dong10,elgamal08,Dennis2011JSAC,swindle11}.
In cooperative jamming, whenever a node transmits a message, a
number of cooperative nodes, called {\em jammers}, help the node conceal its
message by transmitting a carefully chosen signal to raise the
background \emph{noise} level and degrade the eavesdropping channels.
Because our general philosophy applies to any physical layer
approach, the framework can be extended to include other forms
of physical layer security.  However, some of the attractive features
of cooperative jamming that motivated us to study this technique include:
\begin{enumerate}
\item Opportunistic techniques~\cite{maurer93,bloch2008} that exploit the time-varying wireless channel may suffer from excessive delays depending on the rate of channel fluctuations. For applications that require security without an excessive delay, active channel manipulation such as cooperative jamming should be adopted. The price to be paid, in this case, is the increased interference due to jamming.

\item Multi-antenna systems can also be used to jam eavesdroppers~\cite{yates07,goel08}. However, the use of multiple antennas on every wireless device may not be feasible due to cost and size (\eg\ wireless sensors). Cooperative jamming is a distributed alternative to multi-antenna systems.

\item Node cooperation, while requiring a more complex physical layer, is incorporated in commercial wireless technologies such as LTE. Thus, we envision that cooperative jamming can be implemented in practice, as was demonstrated in a limited form (single jammer) in~\cite{katabi11}.

\item Anonymous wireless communication is a challenging problem. Cooperative jamming can potentially be utilized for wireless anonymous communication, as it creates confusion for wireless localization techniques~\cite{banerjee11}.
\end{enumerate}

In this general case, the main questions are: (1) how to choose the
intermediate nodes that form a multi-hop path from the source node to
the destination node, and (2) how to configure each hop at the
physical layer with respect to the security and throughput
constraints of the path. Specifically, the problem we consider in
this paper is how to find a {\em minimum cost} path between a source
and destination node in the network, while guaranteeing a
pre-specified lower bound on the {\em end-to-end secrecy} and {\em
goodput} of the path. The cost of a path can be defined in terms of
various system parameters. In a wireless network, transmission power
is a critical factor affecting the throughput and lifetime of the
network. While increasing the transmission power results in increased
link throughput, excessive power actually results in high levels of
interference, hence reducing the network throughput due to
inefficient spacial reuse. With cooperative jamming at the physical
layer, transmission power is even more important due to the
additional interference caused by jamming signals if they need to be
employed. Thus, in this work, we consider the amount of end-to-end
transmission power as the cost of a path with the objective of
finding secure paths that consume the least amount of energy. In
turn, such paths, by minimizing interference in the network, result
in \emph{higher} throughput.  Note that solutions employing power
only at the nodes transmitting the messages (and no cooperative
jamming) are part of the space over which the optimization will be
performed; thus, if it is more efficient to not employ cooperative
jamming, such a solution will be revealed by our algorithms.

While it might seem that physical layer security techniques can be
extended to multi-hop networks by implementing them on a
hop-by-hop basis, in general, such extensions sacrifice performance
or are not feasible. The eavesdropping probability on a link
is a function of the power allocation on that link.  A hop-by-hop
implementation is unable to determine the optimal eavesdropping
probability and consequently power allocation for each link in
order to satisfy the end-to-end constraints (\ie\ the chicken-egg problem).
Moreover, a hop-by-hop approach overlaid on a shortest path routing algorithm
might pay an enormous penalty to mitigate eavesdroppers on some links (\eg\
by routing through a node with one or more links, that, because
of system geometry, are very vulnerable to nearby eavesdroppers).
A routing algorithm that is designed in conjunction with
physical layer security can selectively employ links that are
easier to secure when it is power-efficient to do so and, in
such a way, minimize the impact of the security constraint on
end-to-end throughput.

Our main contributions can be summarized as follows:
\begin{itemize}
\item We formulate the secure minimum energy routing problem with end-to-end security and goodput constraints as a constrained
optimization of transmission power at the physical layer and link
selection at the network layer.

\item We prove that the secure minimum energy routing problem is \mbox{NP-hard}, and develop
exact and $\epsilon$-approximate solutions of, respectively, pseudo-polynomial and fully-polynomial time complexity for the problem.

\item We show how cooperative jamming can be used to establish a secure link between two nodes in the presence of multiple eavesdroppers or probabilistic information about potential eavesdropping locations by utilizing random linear coding at the network layer.

\item We provide simulation results that demonstrate the significant energy savings
of our algorithms compared to the combination of security-agnostic minimum energy
routing and physical layer security.
\end{itemize}

The rest of the paper is organized as follows. Our system model is
described in Section~\ref{sec:sysmodel}. The optimal link and path
cost are analyzed in Sections~\ref{sec:linkcost}
and~\ref{sec:pathcost}. Our routing algorithms are presented in
Section~\ref{sec:algorithms}. Simulation results are discussed in
Section~\ref{sec:simulation}. Section~\ref{sec:related} presents an
overview of some related work, while Section~\ref{sec:conclusion}
concludes the paper.

\section{System Model and Assumptions}
\label{sec:sysmodel}

Consider a wireless network with arbitrarily distributed nodes.
We assume that each node (legitimate or eavesdropper) is equipped
with a single omni-directional antenna.
A $K$-hop route $\Pi$ between a source and a destination in the
network is a sequence of $K$ links connecting the source to the
destination\fnote{Terms ``path'' and ``route'' are used
interchangeably throughout the paper.}. We use the notation $\Pi =
\angles{\ell_1, \ldots, \ell_K}$ to refer to a route that is formed
by $K$ links $\ell_1$ to $\ell_K$. A link $\ell_k \in \Pi$ is formed
between two nodes $S_k$ and $D_k$ on route $\Pi$. We assume that
every link $\ell_k$ is exposed to a set of (potential) eavesdroppers
denoted by $\mc{E}_k$. Whenever $S_k$ transmits a message to $D_k$,
a set of trusted nodes, called jammers, cooperate with $S_k$ to
conceal its message from the eavesdroppers in $\mc{E}_k$ by jamming
$S_k$'s signal at the eavesdroppers. The set of the jammers
cooperating with $S_k$ to conceal its transmissions from $\mc{E}_k$
is denoted by $\mc{J}_k = \set{J_1\ldots, J_{|\mc{J}_k|}}$, where $|\mc{A}|$
denotes the cardinality of set $\mc{A}$. The set of
jammers is potentially different for different links. Throughout the
paper, we use the notation $\enc{S_k, D_k, \mc{E}_k,
\mc{J}_k}$ to identify link $\ell_k$.

In the following subsections, we describe the models considered in
this paper for the wireless channel, eavesdroppers, physical-layer
security and end-to-end routing. For notational simplicity, we may
drop the link index $k$ whenever there is no ambiguity.

\subsection{Wireless Channel Model}
Consider the discrete-time equivalent model for a transmission from node
$S$ to node $D$. Let $x_S$ be the normalized (unit-power) symbol
stream to be transmitted by $S$, and let $y_D$ be the received
signal at node $D$. We assume that transmitter $S$ is able to
control its power $P_S$ in arbitrarily small steps, up to some
maximum power $\pmax$. Let $n_D$ denote the receiver noise at $D$, where
$n_D$ is assumed to be a complex Gaussian random variable with
$\mathbb{E}\left[|n_D|^2\right] = N_0$. The received signal at $D$
is expressed as
\begin{equation}
    y_D = \sqrt{P_S} \, h_{S,D} \, x_S + n_D,
\end{equation}
where $h_{S,D}$ is the complex channel gain between $S$ and $D$. The
channel gain is modeled as $h_{S,D} = |h_{S,D}|e^{\theta_{S,D}}$,
where $|h_{S,D}|$ is the channel gain magnitude and $\theta_{S,D}$
is the uniform phase. We assume a non line-of-sight environment, implying
that $|h_{S,D}|$ has a Rayleigh distribution, and that
$\mathbb{E}[|h_{S,D}|^2] = 1/d_{S,D}^\alpha$, where  $d_{S,D}$ is the
distance between nodes $S$ and $D$,
and $\alpha$ is the path-loss exponent (typically between $2$ and
$6$). This is the standard narrowband fading channel model
employed in the physical layer literature~\cite{tse05wireless,goldsmith2005}.

\subsection{Adversary Model}
We limit our attention to passive eavesdroppers as in prior
work~\cite{goel08,tekin2008general,dong10,elgamal08,Dennis2011JSAC,swindle11}.
Although there are other forms of adversarial behavior, their
consideration is beyond the scope of this paper.

While the literature on physical layer security often assumes not
only eavesdropper locations but also either perfect
(\eg~\cite{tekin2008general}) or imperfect (\eg~\cite{swindle11})
knowledge at the transmitters and jammers of the complex channel
gains of the eavesdropping channels (\ie\ availability of
instantaneous eavesdropper channel state information (CSI)), we
consider the more realistic scenario, in which CSI for eavesdropping
channels is not available. In addition, our model requires only the
knowledge of {\em potential} eavesdropping locations in the network,
yet we show that it provides guaranteed security by employing coding
in conjunction with cooperative jamming.

Specifically, we assume that each link $\ell_k$ is subject to
potential eavesdropping from a set of locations denoted by $\mc{E}_k
= \set{E_1, \ldots, E_{|\mc{E}_k|}}$, where the probability of
eavesdropping from location $E_i$ is given by $p(E_i)$ for $0 \le
p(E_i) \le 1$. This is a considerably general model that can be used
to represent a wide range of eavesdropping scenarios\fnote{Although
our model cannot be applied to every possible scenario, it is more
general compared to the models in the literature on physical layer
security
(see~\cite{goel08,tekin2008general,dong10,elgamal08,Dennis2011JSAC,swindle11},
and references therein).}. For example, setting all $p(E_i)$'s to $1$
for a link models multiple eavesdroppers for that link.  Other
examples include, for example, military scenarios where the locations
of enemy installations are known, or wireless networks where a
malicious user(s) has been detected.  In general, for any given link,
there is only a limited region around the link that can be exploited
for eavesdropping. By dividing the effective eavesdropping region to
a few smaller areas~\cite{tessellation}, one can compute the most
effective eavesdropping location within each area, and consequently,
construct the set of eavesdropping locations for that link.

\subsection{Physical Layer Security Model}
Consider a secure link formed between source $S$ and receiver $D$
with the help of jammers $\mc{J}$.  For the moment, we assume that
cooperative jamming is implemented at the physical layer to deal with
a \emph{single} eavesdropper $E$ located at a \emph{fixed} position.
Later, in Section~\ref{sec:linkcost}, we show how this physical layer
primitive can be used to provide security against multiple
eavesdroppers or unknown eavesdropping locations.

When node $S$ transmits a message, there are multiple ways in which
cooperative jamming by system nodes can be exploited, ranging from
relatively simple noncoherent techniques to sophisticated beamforming
techniques~\cite{mostafa_twireless}.  Since the {\em implementation}
of beamforming in other contexts, with the same challenges of
synchronization in the wireless environment, is advancing rapidly
\cite{beamform_overview,beamform_implementation}, we assume that the
jammers cooperatively \emph{beamform} a common artificial noise
signal $z$ to the receiver in such a way that their signals cancel
out at the receiver \cite{brown_null}.  The noise signal $z$ is
transmitted in the \emph{null~space} of the channel vector
    $\mx{h}_D = [h_{J_1,D}, h_{J_2,D}, \ldots, h_{J_{|\mc{J}|},D}]^\mathrm{T}$
where, $h_{J_i,D}$ denotes the channel gain between jammer $J_i$ and
destination $D$ and $\mx{A}^\mathrm{T}$ denotes the conjugate transpose of vector
$\mx{A}$. Thus, the signal transmitted by the jammers can be
expressed as
    $\mathbf{s}_J = \mathbf{h}_D^\perp \, z$,
where $\mathbf{h}_D^\perp$ is a vector chosen in the null
space of $\mathbf{h}_D$. It follows that the total transmission
power of the jammers is given by $P_J = \parallel\mathbf{h}_D^\perp\parallel^2$.
Assuming that the source node transmits with power $P_S$, the
signals received at the destination and the eavesdropper are given
by
\begin{equation}
\begin{split}
\label{eq:rsignals}
    y_D &= \sqrt{P_S} \, h_{S,D} \, x_S + n_D,\\
    y_E &= \sqrt{P_S} \, h_{S,E} \, x_S + \mathbf{h}_E^\mathrm{T} \mathbf{h}_D^\perp z + n_E,
\end{split}
\end{equation}
where, $\mathbf{h}_E = [h_{J_1,E}, h_{J_2,E}, \ldots,
h_{J_{|\mc{J}|},E}]^\mathrm{T}$ represents the channel gain vector between
the jammers and the eavesdropper, and $n_D$ and $n_E$ denote the
complex Gaussian noise at the destination and eavesdropper,
respectively, with $\E{|n_D|^2} = \E{|n_E|^2} = N_0$.

Although the jammers try to prevent the eavesdropper from
successfully receiving the message, there is still some probability
that the eavesdropper actually obtains the message due to the fact
that the channel to the eavesdropper is \emph{unknown} in our model, \ie\ $\mx{h}_E$ and
$h_{S,E}$ are unknown. Recalling that the signal-to-interference-plus-noise
ratio ($\sinr$) at the destination is controlled via power control, let
$\gamma_E$ denote the minimum required $\sinr$ at the eavesdropper
in order to violate the security constraints of the protocol
(e.g. for the cryptographic case, the $\sinr$ above which the eavesdropper
can record a meaningful version of the transmitted signal; in the
information-theoretic case, the $\sinr$ above which, for a given wire-tap code,
the equivocation does not equal the entropy of the message.)  Let $\sinr_E$
denote the $\sinr$ at the eavesdropper.  We have
\begin{equation}
\label{eq:success_beaming}
\begin{split}
  &\pr{ \sinr_E \ge \gamma_E}
    =
      \pr{
          \tfrac{ P_S|h_{S,E}|^2 }{ N_0 + \mathbf{h}_E^\mathrm{T}  \mathbf{h}_D^\perp {\mathbf{h}_D^\perp}^\mathrm{T} \mathbf{h}_E }
          \ge \gamma_E }\\
    &\qquad= \mathbb{E}_{\mathbf{h}_E}
        \left[
            \pr{
                \frac{ P_S|h_{S,E}|^2 }{ N_0 + \mathbf{h}_E  \mathbf{h}_D^\perp {\mathbf{h}_D^\perp}^\mathrm{T} \mathbf{h}_E^\mathrm{T} }
                \ge \gamma_E \Big|\, \mathbf{h}_E  }
        \right] \\
    &\qquad= \mathbb{E}_{\mathbf{h}_E}
        \big[
            e^{- \frac{\gamma_E d_{S,E}^\alpha}{P_S} \mathbf{h}_E^\mathrm{T} \mathbf{h}_D^\perp {\mathbf{h}_D^\perp}^\mathrm{T} \mathbf{h}_E }
        \big]
        e^{- \frac{\gamma_E N_0 d_{S,E}^\alpha}{P_S}} ,
 \end{split}
\end{equation}
where $\mathbb{E}_{\mathbf{h}_E}$ means expectation with respect to
channel gain vector $\mathbf{h}_E$. Using the results on quadratic forms~\cite[Eq.~14]{turin} to calculate the expectation, it is obtained that ($\mathbf{I}_N$ is
the identity matrix of size $N$)
\begin{equation}
\label{eq:success}
\begin{split}
  \pr{ \sinr_E \ge \gamma_E}
    &\le \frac{ e^{- \frac{\gamma_E N_0 d_{S,E}^\alpha}{P_S}} }
        { |\mathbf{I}_N + \frac{\gamma_E d_{S,E}^\alpha}{P_S}(\sum_{J_i \in \mc{J}}{\frac{1}{d_{J_i,E}^\alpha}}) \mathbf{\mathbf{h}_D^\perp {\mathbf{h}_D^\perp}^\mathrm{T}}| } \\
    &= \frac{ e^{ -N_0 \gamma_E \frac{d_{S,E}^\alpha}{P_S} } }
            { 1 + \frac{\gamma_E d_{S,E}^\alpha}{P_S}(\sum_{J_i \in \mc{J}}{\frac{1}{d_{J_i,E}^\alpha}}) P_J }
            \eqend
\end{split}
\end{equation}
where the final expression is derived from Sylvester's determinant
theorem:
\[\det(\mathbf{I}_m + \mathbf{A}\mathbf{B}) = \det(\mathbf{I}_n + \mathbf{B}\mathbf{A} ),\]
for $\mathbf{A}$ and $\mathbf{B}$ being $m\times n$ and $n\times m$
matrices, respectively, and the fact that $P_J =
{\mathbf{h}_D^\perp}^\mathrm{T} \mathbf{h}_D^\perp$ (see
\eqref{eq:pj1}).

In the remainder of the paper, we use~\eqref{eq:success} in equality form to compute the eavesdropping probability for a given jamming power $P_J$. While this results in a (slightly)
conservative power allocation, it is sufficient to satisfy the security requirement of
each link.

\subsection{Routing Model}
Consider a $K$-hop route $\Pi = \angles{\ell_1, \ldots, \ell_K}$ between a
legitimate source and destination in the network. Let $\mc{L}$ denote the set
of all possible routes between the source and destination. Let
$\mc{C}(\Pi)$ denote the cost of route $\Pi$, where the cost of a
route is defined as the summation of the costs of the links forming the
route. With slight abuse of the notation, we use $\mc{C}(\ell_k)$ to
denote the cost of link $\ell_k$ as well. The secure routing problem
is then defined as follows.

\bigskip\noindent {\bf SMER: Secure Minimum Energy Routing Problem}\\
{\em Consider a wireless network and a set of eavesdroppers
distributed in the network. Given a source and destination, find a
minimum energy path $\Pi^*$ between the source and destination
subject to constraints $\pi$ and $\lambda$ on the end-to-end successful eavesdropping
probability and goodput on the path respectively.}
\bigskip

Let $\lambda(\Pi)$ and $\lambda(\ell_k)$ denote, respectively, the goodput
of path $\Pi$ and link $\ell_k \in \Pi$. Then $\lambda(\Pi)$ can be expressed as
    \[ \lambda(\Pi) = \min_{\ell_k \in \Pi} \lambda(\ell_k) \eqend\]
Since goodput of a link is an increasing function of the transmission
power of the transmitter of that link, a necessary condition for minimizing power over the path $\Pi$
is given by $\lambda(\ell_k) = \lambda$, for all $\ell_k \in \Pi$, \ie\
all links should just achieve the minimum goodput $\lambda$. Thus, our power allocation scheme (see Section~\ref{sec:linkcost}) establishes links that achieve exactly the minimum required goodput $\lambda$. Consequently, the constraint on the end-to-end goodput is satisfied by any path in the network, and hence does not need to be explicitly considered when solving SMER.
As such, SMER can be formally described by the following optimization problem:
\begin{equation}
\begin{split}
\label{eq:smer}
    \Pi^* = \, &\displaystyle\arg \min_{\Pi \in \mc{L}}  \sum_{\ell_k \in \Pi} \mc{C}(\ell_k) \\
                &s.t. \quad \pr{\text{eavesdropping on route $\Pi$}} \le \pi,
\end{split}
\end{equation}
for some pre-specified $\pi$ ($0 < \pi < 1$). The constraint on
the route eavesdropping probability in the above optimization
problem can be expressed in terms of the eavesdropping probability
on individual links $\ell_k$ that form the route $\Pi$, as
        $\prod_{\ell_k \in \Pi} (1 - \pi_k) \ge 1-\pi$,
where $\pi_k$ ($ 0 < \pi_k < 1$) denotes the successful
eavesdropping probability on link $\ell_k$. We use the following
result to convert the above inequality constraint to an equality
constraint in the routing problem~\eqref{eq:smer}.

\begin{lemma}
The cost of route $\Pi$ is a monotonically
increasing function of
$\prod_{\ell_k \in \Pi} (1 - \pi_k)$.
\end{lemma}
\begin{proof}
Consider a path $\Pi$ between the source and destination nodes.
Define the end-to-end secrecy of path $\Pi$, denoted by
$\omega(\Pi)$, as follows:
\begin{equation}
    \omega(\Pi) = \prod_{\ell_k \in \Pi} \omega_k,
\end{equation}
where $\omega_k = (1 - \pi_k)$.

First, we show that the link cost $\mc{C}(\ell_k)$ is a
monotonically increasing function of the the link secrecy $\omega_k$
for every link $\ell_k \in \Pi$. Let $P_S^{(k)}$ and $P_J^{(k)}$
denote the source and jamming powers allocated to the link $\ell_k$,
respectively. In Section~\ref{sec:linkcost}, we show that: (i)
$P_S^{(k)}$ is a constant for a given link independent of the link
secrecy, and (ii) $P_J^{(k)}$ is a function of the link secrecy and
is given by
\begin{equation}
    P_J^{(k)} = c_k \cdot \frac{\omega_k}{1 - \omega_k},
\end{equation}
where $c_k$ is some constant independent of $\omega_k$. Thus, for a
fixed link $\ell_k$, the link cost $\mc{C}(\ell_k) = P_S^{(k)} +
P_J^{(k)}$ depends on $\omega_k$ only through the jamming power
$P_J^{(k)}$. Taking the derivative on the link cost with respect to
$\omega_k$ results in the following relation:
\begin{equation}
    \frac{d}{d \omega_k} \, \mc{C}(\ell_k) = c_k \cdot \frac{1}{(1 - \omega_k)^2} > 0,
\end{equation}
indicating that the link cost is an increasing function of the link
secrecy $\omega_k$.

Let $\mc{C}^*(\Pi)$ and $\mc{C}(\Pi)$ denote the optimal cost of the
path $\Pi$ computed by solving the optimization
problem~\eqref{eq:smer-p}, with equality and inequality constraints,
respectively. Furthermore, let $\omega^*(\Pi)$ and $\omega(\Pi)$
denote the corresponding end-to-end path secrecies. We present a
proof of the lemma based on contradiction by assuming that the
optimal path cost with the inequality constraint is less than that
with the equality constraint. That is, we assume that
\begin{equation}
    \mc{C}(\Pi) \le \mc{C}^*(\Pi),
\end{equation}
while,
\begin{equation}
    \omega(\Pi) > \omega^*(\Pi) \eqend
\end{equation}

Next, by manipulating the link secrecy allocation vector $[\omega_1,
\ldots, \omega_k, \ldots, \omega_K]$, we construct a new link
secrecy allocation vector that satisfies the equality constraint,
while having a cost smaller than $\mc{C}^*(\Pi)$. To this end,
consider some arbitrary link $\ell_k$, and replace $\omega_k$ by a
new $\omega'_k$ as follows
\begin{equation}
    \omega'_k = \frac{\omega^*}{\omega} \cdot \omega_k
    \eqend
\end{equation}
Since $\omega^* < \omega$, it follows that $\omega'_k < \omega_k$.
Consequently, the new cost of the link $\ell_k$ with link secrecy
$\omega'_k$ is less than $\mc{C}(\ell_k)$, which in turn indicates
that the new path cost with secrecy allocation vector $[\omega_1,
\ldots, \omega'_k, \ldots, \omega_K]$ is less than $\mc{C}(\Pi)$.
Therefore, we have
\begin{equation}
    \mc{C}'(\Pi) < \mc{C}(\Pi) < \mc{C}^*(\Pi),
\end{equation}
and,
\begin{equation}
    \omega(\Pi) > \omega'(\Pi) = \omega^*(\Pi) \eqend
\end{equation}
The proof follows by noting that this contradicts the assumption
that $\mc{C}^*(\Pi)$ is the minimum cost of path $\Pi$ with the
equality constraint.
\end{proof}

Thus, to minimize the cost of the optimal route, the inequality
constraint can be substituted by the equality constraint
$\prod_{\ell_k \in \Pi} (1 - \pi_k) = 1-\pi$. On each link $\ell_k$,
it is desirable to keep the successful eavesdropping probability
$\pi_k$ close to $0$. In this case, the product $\prod_{\ell_k \in
\Pi} (1 - \pi_k)$ can be approximated by the expression $1 -
\sum_{\ell_k \in \Pi} \pi_k$. By substituting the approximate
linearized constraint in the routing problem, the following
optimization problem is obtained
\begin{equation}
\begin{split}
\label{eq:smer-p}
    \Pi^* = \, &\displaystyle\arg \min_{\Pi \in \mc{L}}  \sum_{\ell_k \in \Pi} \mc{C}(\ell_k) \\
                     &s.t. \quad \displaystyle \sum_{\ell_k \in \Pi}
                     \pi_k = \pi
                     \eqend
\end{split}
\end{equation}
In the rest of the paper, we focus on
this optimization \mbox{problem}.
We first show that the problem is, in general,
NP-hard and then develop exact and approximate algorithms to
solve it.

\section{Secure Link Cost}
\label{sec:linkcost}

The link cost is composed of two components: (1) the source power, and (2)
the jammers' power. Let $\mc{C}(\ell_k)$ denote the cost of link
$\ell_k=\enc{S_k, D_k, \mc{E}_k, \mc{J}_k}$ under the constraint of
eavesdropping probability $\pi_k$. Then, $\mc{C}(\ell_k)$ is given by:
\begin{eqnarray}
    \mc{C}(\ell_k) = P_S^{(k)} + P_J^{(k)},
\end{eqnarray}
where $P_S^{(k)}$ and $P_J^{(k)}$ denote, respectively, the average
source and jammers power on link $\ell_k$. In the following
subsections, we will compute the optimal values of $P_S^{(k)}$ and
$P_J^{(k)}$ subject to a given $\pi_k$.

\subsection{Source Transmission Power}
Assume that the (complex) fading channel coefficient $h_{S_k,D_k}$ is
known at the source $S_k$ of the given link $\ell_k$. Because we are
trying to maintain a fixed rate (and, hence, a fixed received power),
the source will attempt to invert the channel using power control.
However, for a Rayleigh frequency-nonselective fading channel, as
assumed here, the expected required power for such an inversion goes
to infinity, and, hence {\em truncated channel inversion} is employed
\cite[Pg.~112]{goldsmith2005}. In truncated channel inversion, the
source maintains the required link quality except for extremely bad
fades, where the link goes into outage. When a link is in a bad fade,
the source will need to wait until the link improves before
transmitting the packet and delay will be incurred.  To limit the
delay, we maintain a given outage probability $\rho$ per link. Then,
for a given packet, we need to transmit at rate $R = \lambda / (1 -
\rho)$ to maintain the desired goodput $\lambda$. Associated with
that rate $R$ is the SINR threshold $\gamma_D = 2^{R}- 1$ required
for successful reception at the link
destination~\cite{tse05wireless}.

Let $P_S^{(k)}$ denote the average transmission power of $S_k$, and
let $P_S^{(k)}(|h_{S_k,D_k}|^2)$ denote the power used for a given
packet as a function of the power $|h_{S_k,D_k}|^2$ in the fading
channel between $S_k$ and $D_k$.  Per above, below some threshold
$\tau$, the source will wait for a better channel.  From the
Rayleigh fading model employed,
$|h_{S_k,D_k}|^2$ is exponential with parameter
$1/d_{S_k,D_k}^{\alpha}$; hence, $\tau = - \ln(1 - \rho) \cdot d_{S_k,D_k}^{\alpha}$ and truncated channel inversion yields:
\begin{equation}
    P_S^{(k)}(|h_{S_k,D_k}|^2) = \begin{cases}
                                    \frac{\gamma_D}{|h_{S_k,D_k}|^2} \cdot d_{S_k, D_k}^{\alpha},
                                             &       |h_{S_k,D_k}|^2 \geq \tau \\
                                     0, &       |h_{S_k,D_k}|^2 < \tau
                                 \end{cases}
\end{equation}
Then, the average power employed on the link is given by:
\begin{eqnarray}
\label{eq:ps}
    P_S^{(k)} & = & \frac{1}{1 - \rho} \int_{\tau}^{\infty} \frac{\gamma_D}{x} \cdot d_{S_k,D_k}^{\alpha}
            e^{-x} dx   \nonumber\\
          & = & \gamma_D d_{S_k,D_k}^{\alpha} \frac{1}{1 - \rho} \int_{\tau}^{\infty} \frac{e^{-x}}{x} dx \nonumber\\
          & = & \gamma_D \, k_{\rho} \, d_{S_k,D_k}^{\alpha},
\end{eqnarray}
where $k_{\rho}$ is a constant that depends on the parameter
${\rho}$.  Hence, for a fixed network parameter $\rho$ (which also
determines $\gamma_D$), the average power consumed on a given link
$\ell_k$ to achieve the secure goodput $\lambda$ is proportional to
$d_{S_k, D_k}^{\alpha}$.

\subsection{Jammers' Transmission Power}
Our physical layer security primitive described in
Section~\ref{sec:sysmodel} can provide security only against a single
eavesdropper at a fixed location. To achieve security in the presence
of multiple eavesdroppers or uncertainty about the location of
eavesdroppers, we utilize \emph{random linear coding}\fnote{Other
forms of coding, such as dividing each message to smaller
chunks~\cite{shamir79}, can be equally incorporated in our
algorithm.} on each link.

Consider link $\ell_k$ between transmitter $S_k$ and receiver $D_k$
with the associated set of potential eavesdropping locations
$\mc{E}_k=\set{E_1, \ldots, E_{|\mc{E}_k|}}$. Transmitter $S_k$
performs coding over $|\mc{E}_k|$ messages accumulated in its buffer
for transmission to $D_k$. To generate a coded message, $S_k$ selects
a random subset of the messages in its buffer and adds them together
(module-$2$). To recover the original messages, the receiver needs to
collect $|\mc{E}_k|$ linearly independent coded messages. In order to
transmit only linearly independent coded messages, $S_k$ keeps track
of the coded messages it has transmitted so far. Let $m_i$ denote the
$i$-th coded message that is being transmitted to $D_k$. To securely
transmit $m_i$, $S_k$ employs the cooperative jamming primitive of
Section~\ref{sec:sysmodel} assuming that there is an eavesdropper in
location $E_i$. Since each coded message is hidden from at least one
eavesdropping location, it is guaranteed that an eavesdropper located
at location $E_i$, for all $E_i \in \mc{E}_k$,  will not be able to
obtain any information about the original messages.

In the following subsections, we compute the optimal jamming power
per link. The derivation for the case of multiple eavesdroppers
relies on the jamming power computed for the single eavesdropper
case.

\subsubsection{Single Eavesdropper}
Because slow frequency non-selective fading is assumed and
the channel to the eavesdropper is unknown, there is some probability
that the eavesdropper will obtain the message by achieving a received
SINR greater than a threshold $\gamma_E$.
Let $\pi_k(|h_{S_k,D_k}|^2)$ denote the probability the eavesdropper
achieves SINR greater than threshold $\gamma_E$ for a given source to
destination channel $h_{S_k,D_k}$ (recall that the source power will
fluctuate as $h_{S_k,D_k}$ fluctuates, and this will impact the
interception probability at the eavesdropper). Because we want to
avoid placing limitations on the capabilities of the eavesdropper,
assume that the eavesdropper receiver is noiseless. Let $P_J^{(k)}$
and $P_J^{(k)}(|h_{S_k,D_k}|^2)$ denote the average and instantaneous
transmission power allocated to jammers in $\mc{J}_k$, respectively.
Then, using~\eqref{eq:success}, it is obtained that
\begin{equation*}
    \pi_k(|h_{S_k,D_k}|^2) =
    \quad \frac{1}
        {1 + \frac{\gamma_E d_{S_k, E_k}^{\alpha}}
    {P_S^{(k)}(|h_{S_k,D_k}|^2)}
        \big( \sum_{J_i \in \mc{J}_k} \frac{1}{d_{J_i, E_k}^{\alpha}} \big)
    P_J^{(k)}(|h_{S_k,D_k}|^2)} \eqend
\end{equation*}
Now, to maintain a given $\pi_k$, it is sufficient to
maintain $\pi_k(|h_{S_k,D_k}|^2) = \pi_k$ across all
$|h_{S_k,D_k}|^2$. Under this condition, recognizing that both
$P_S^{(k)}(|h_{S_k,D_k}|^2)$ and $P_J^{(k)}(|h_{S_k,D_k}|^2)$ are
proportional to $|h_{S_k,D_k}|^2$, we have:

\begin{equation}
\begin{split}
    P_J^{(k)}(|h_{S_k,D_k}|^2) =
    \frac{(1/\pi_k - 1) \, P_S^{(k)}(|h_{S_k,D_k}|^2)}
        {\gamma_E \, d_{S_k,E_k}^{\alpha} (\sum_{J_i \in \mc{J}_k} \frac{1}{d_{J_i, E_k}^{\alpha}})}
         ,
\end{split}
\end{equation}
and, taking expectations yields
\begin{equation}
\label{eq:P_J}
    P_J^{(k)} =  \frac{1/\pi_k - 1}
        {\gamma_E d_{S_k,E_k}^{\alpha} (\sum_{J_i \in \mc{J}_k} \frac{1}{d_{J_i,E_k}^{\alpha}})}
        P_S^{(k)} ,
\end{equation}
and,
\begin{equation}
\label{eq:pi1}
    \pi_k =  \frac{1}
        {1 + \frac{\gamma_E d_{S_k,E_k}^{\alpha}}{P_S^{(k)}} \big(\sum_{J_i \in \mc{J}_k} \frac{1}{d_{J_i, E_k}^{\alpha}}\big)P_J^{(k)} }
        \eqend
\end{equation}

\subsubsection{Multiple Eavesdroppers}
Recall that our objective is to compute the minimum jamming power for
the link.
Let $\pi_k(i)$ denote the successful eavesdropping probability on
link $\ell_k$ conditioned on having an eavesdropper at location
$E_i$. The unconditional eavesdropping probability $\pi_k$ on link
$\ell_k$ is then given by the approximate relation $\pi_k =
\sum_{E_i \in \mc{E}_k} p_k(E_i) \cdot \pi_k(i)$, where $p_k(E_i)$ is
the probability of having an eavesdropper at location $E_i$. Since
jamming power depends on the location of the eavesdroppers, by
optimally allocating jamming power to each potential eavesdropping
location, we can minimize the total jamming power across all
eavesdropping locations for a given link.

The minimum jamming power for link $\ell_k$ over all eavesdropping
locations $\mc{E}_k$ is given by the solution of the following
optimization problem:
\begin{equation}
\begin{split}
     &\min_{P_J^{(k)}(i)}  \sum_{E_i \in \mc{E}_k} P_J^{(k)}(i) \\
     &s.t. \quad
        \displaystyle\sum_{E_i \in \mc{E}_k} p_k(E_i) \cdot \pi_k(i) = \pi_k,
\end{split}
\end{equation}
where $P_J^{(k)}(i) = \sum_{J_j \in \mc{J}_k} P_j^{(k)}(i)$ is the
jamming power conditioned on the eavesdropping location $E_i$, \ie\
the jamming power during the transmission of the coded message $m_i$.
Define $\phi_k(i)$ as follows
\begin{equation}
    \phi_k(i) =  \frac{\gamma_E}{\gamma_S k_{\rho}} \big(\frac{d_{S_k, E_i}}{d_{S_k, D_k}}\big)^{\alpha} \sum_{J_j \in \mc{J}_k} \frac{1}{d_{J_j, E_i}^{\alpha}}
    \eqend
\end{equation}
After substituting for $\pi_k(i)$ using~\eqref{eq:pi1}, we obtain the
following optimization problem:
\begin{equation}
\label{e:multi}
\begin{split}
     &\min_{P_J^{(k)}(i)}  \sum_{E_i \in \mc{E}_k} P_J^{(k)}(i) \\
     &s.t. \quad
        \displaystyle\sum_{E_i \in \mc{E}_k}
            \frac{p_k(E_i)}
            {1 + \phi_k(i) P_J^{(k)}(i) }
        = \pi_k
        \eqend
\end{split}
\end{equation}
The optimization variables in this optimization problem are the
jamming powers $P_J^{(k)}(i)$. 
The Lagrangian for the link cost
optimization problem is expressed as follows
\begin{equation*}
\begin{split}
    &L(P_J^{(k)}(1), \ldots, P_J^{(k)}(|\mc{E}_k|), \nu)\\
        & = \sum_{E_i \in \mc{E}_k}
                P_J^{(k)}(i)
            + \nu
                \Big(
                    \displaystyle\sum_{E_i \in \mc{E}_k}
                    \frac{p_k(E_i)}
                    {1 + \phi_k(i) P_J^{(k)}(i) }   
                    - \pi
                \Big)
                \eqend
\end{split}
\end{equation*}
Using the Lagrange multipliers technique, it is obtained that
\begin{equation}
\label{eq:LL1}
    \frac{\partial L}{\partial P_J^{(k)}(i)}
        = 1 - \nu
                \frac{\phi_k(i) p_k(E_i)}
                    {(1 + \phi_k(i) P_J^{(k)}(i))^2} ,
\end{equation}
and,
\begin{equation}
\label{eq:LL2}
    \frac{\partial L}{\partial \nu}
        = \sum_{E_i \in \mc{E}_k}
                \frac{p_k(E_i)} {1 + \phi_k(i) P_J^{(k)}(i) }
           - \pi
           \eqend
\end{equation}
Using \eqref{eq:LL1}, we have
\begin{equation}
\label{eq:LL1-s}
    \frac{p_k(E_i)} {1 + \phi_k(i) P_J^{(k)}(i) }
    =
    \frac{\sqrt{p_k(E_i)}}{\sqrt{\nu \phi_k(i)}}
    \eqend
\end{equation}
By substituting in~\eqref{eq:LL2}, it follows that
\begin{equation}
    \sum_{E_i \in \mc{E}_k} \sqrt{\frac{p_k(i)}{\nu \phi_k(i)}} = \pi ,
\end{equation}
and, therefore,
\begin{equation}
   \frac{1}{\sqrt{\nu}}
        = \frac{\pi}{\sum_{E_i \in \mc{E}_k} \sqrt{\frac{p_k(i)}{\phi_k(i)}} }
        \eqend
\end{equation}
It is then obtained that
\begin{equation}
    \pi_k(i) = \frac{1}{\phi_k(i)} \frac{1/\sqrt{\frac{p_k(E_i)}{\phi_k(i)}}}
                 { \sum_{E_i \in \mc{E}_k} \sqrt{\frac{p_k(E_i)}{\phi_k(i)}}} \, \pi_k,
\end{equation}
and,
\begin{equation}
\label{eq:pji}
   P_J^{(k)}(i) =
            \frac{1}{\pi_k} \sqrt{\frac{p_k(E_i)}{\phi_k(i)}} \sum_{E_i \in \mc{E}_k} \sqrt{\frac{p_k(E_i)}{\phi_k(i)}} - \frac{1}{\phi_k(i)}
            \eqend
\end{equation}
For a given link $\ell_k$ and eavesdropping probability $\pi_k$, we
can use~\eqref{eq:pji} to compute the optimal jamming power
allocation for each coded message $m_i$. Consequently, the average
jamming power per message on link $\ell_k$ is given by:
\begin{equation}
\label{eq:pj1}
\begin{split}
   P_J^{(k)} &= \frac{1}{|\mc{E}_k|} \sum_{E_i \in \mc{E}_k} P_J^{(k)}(i)\\
             &= \frac{1}{\pi_k} \frac{1}{|\mc{E}_k|} \Bigg( \sum_{E_i \in \mc{E}_k} \sqrt{\frac{p_k(E_i)}{\phi_k(i)}} \Bigg)^2
                - \frac{1}{|\mc{E}_k|} \sum_{E_i \in \mc{E}_k} \frac{1}{\phi_k(i)}
        \eqend
\end{split}
\end{equation}

\subsection{Discussion}
\subsubsection{Colluding Eavesdroppers}
While we considered the case of non-colluding
eavesdroppers here, our model can be extended to handle
colluding eavesdroppers by requiring that at least of the coded messages be protected against all
eavesdroppers. Let $\mc{E}_k=\set{E_1, \ldots, E_{|\mc{E}_k|}}$ denote the set of colluding eavesdroppers. Assume that on link $\ell_k$, $B_k$ messages are coded together for transmission, \ie\ $B_k$ is the length of the coding block. Then, the probability that a coded message $m$ is captured by all eavesdroppers is given by $\prod_{E_i \in \mc{E}_k} \pi_k(i)$. Thus, the probability that at least one message out of the $B_k$ coded messages is not received by all eavesdroppers is given by
\begin{equation}
    1 - \Big( \prod_{E_i \in \mc{E}_k} \pi_k(i) \Big)^{B_k}
        \eqend
\end{equation}
To satisfy the link eavesdropping constraint $\pi_k$, the following relation should be satisfied
\begin{equation}
    1 - \Big( \prod_{E_i \in \mc{E}_k} \pi_k(i) \Big)^{B_k} = \pi_k,
\end{equation}
which yields
\begin{equation}
    \prod_{E_i \in \mc{E}_k} \pi_k(i) = \sqrt[B_k]{\pi_k}
    \eqend
\end{equation}
This constraint can be used in the optimization problem~\eqref{e:multi} to compute the optimal link cost for the case of colluding eavesdroppers.

An interesting observation is that
\begin{equation}
    \lim_{B_k \rightarrow \infty} \pi_k(i) = 1, \quad \text{for all $E_i \in \mc{E}_k$}
    \eqend
\end{equation}
That is, by increasing the length of the coding block, the link cost can be significantly reduced. The cost to be paid is in terms of increased transmission delay.

\subsubsection{End-to-End Coding}
Rather than looking at individual links in isolation and then
performing hop-by-hop coding, we can perform coding on an end-to-end
basis only at the source node. Then by repeatedly finding paths that
are secure against single eavesdropping per link, the source can
securely communicate with the destination through multiple paths.
This approach is appropriate if there are only a few potential
eavesdropping locations in the network. If the maximum number of
eavesdropping locations per link is $m$, then the running time of
this approach is $m$ times that of the routing algorithm with single
eavesdropping location per link.

\section{Secure Path Cost}
\label{sec:pathcost} In this section, using the link cost formulation
of the previous section, we formulate the optimal cost of a
\emph{given} path $\Pi$ subject to an end-to-end eavesdropping
probability $\pi$.
The problem essentially is to divide $\pi$ across the links forming $\Pi$ so that the path cost is minimized.

\subsection{Optimal Path Cost}
Consider a given path $\Pi$. We find the optimal cost of path $\Pi$
by solving the optimization problem~\eqref{eq:smer-p}. Consider link
$\ell_k \in \Pi$, where $\ell_k = \enc{S_k, D_k, \mc{E}_k,
\mc{J}_k}$. Define $x_k$ and $y_k$ as follows:
\begin{align*}
    x_k = \frac{1}{\sqrt{|\mc{E}_k|}} \sum_{E_i \in \mc{E}_k} \sqrt{\frac{p_k(E_i)}{\phi_k(i)}},
\end{align*}
and,
\begin{align*}
    y_k = \frac{1}{|\mc{E}_k|} \sum_{E_i \in \mc{E}_k} \frac{1}{\phi_k(i)}
    \eqend
\end{align*}
Using the results obtained in the previous subsection, the following
relation holds:
\begin{align*}
    \pi_k =  \frac{x^2_k}{y_k + P_J^{(k)}}
    \eqend
\end{align*}
By substituting the above expressions in the optimal routing
formulation described in~\eqref{eq:smer-p}, the following
optimization problem is obtained for minimizing the cost
$\mc{C}(\Pi)$ of route $\Pi$:
\begin{equation}
\begin{split}
     &\displaystyle\min_{P_J^{(k)}}  \sum_{\ell_k \in \Pi} P_S^{(k)} + P_J^{(k)} \\
     &s.t. \quad
        \displaystyle\sum_{\ell_k \in \Pi}
                        \Big(
                         \frac{x^2_k} {y_k + P_J^{(k)}}
                        \Big)
                            = \pi
                            \eqend
\end{split}
\end{equation}
The optimization variables in this optimization problem are jamming
powers $P_J^{(k)}$. 
The Lagrangian for the routing
optimization problem is expressed as follows
\begin{equation*}
\begin{split}
    &L(P_J^{(1)}, \ldots, P_J^{(K)}, \nu)\\
        & = \sum_{\ell_k \in \Pi}
            \big( P_S^{(k)} + P_J^{(k)} \big)
            + \nu
                \Big(
                    \sum_{\ell_k \in \Pi}
                        \Big(
                         \frac{x^2_k} {y_k + P_J^{(k)}}
                        \Big)
                    - \pi
                \Big)
                \eqend
\end{split}
\end{equation*}
Using the Lagrange multipliers technique, it is obtained that
\begin{equation}
\label{eq:L1}
    \frac{\partial L}{\partial P_J^{(k)}}
        = 1 - \nu
                \frac{x^2_k} {(y_k + P_J^{(k)})^2} ,
\end{equation}
and,
\begin{equation}
\label{eq:L2}
    \frac{\partial L}{\partial \nu}
        = \sum_{\ell_k \in \Pi}
              \Big(
                 \frac{x^2_k} {y_k + P_J^{(k)}}
              \Big)
           - \pi
           \eqend
\end{equation}
Using \eqref{eq:L1}, we have
\begin{equation}
\label{eq:L1-s}
    \frac{x^2_k} {y_k + P_J^{(k)}}
    =
    \frac{x_k}{\sqrt{\nu}}
    \eqend
\end{equation}
By substituting in~\eqref{eq:L2}, it follows that
\begin{equation}
    \sum_{\ell_k \in \ell} \frac{x_k}{\sqrt{\nu}} = \pi ,
\end{equation}
and, therefore,
\begin{equation}
   \frac{1}{\sqrt{\nu}}
        = \frac{\pi}{\sum_{\ell_k \in \Pi} x_k}
        \eqend
\end{equation}
After substitution in~\eqref{eq:L1-s}, the following relation for
the optimal eavesdropping probability $\pi_k$ on link $\ell_k$ is
obtained
\begin{equation}
\label{eq:pi}
    \pi_k = \frac{x_k}
                {\sum_{\ell_i \in \Pi} x_i} \, \pi
                \eqend
\end{equation}
For a given route $\Pi$ and end-to-end eavesdropping probability
$\pi$, we can use \eqref{eq:pi} to divide $\pi$ between links
$\ell_k \in \Pi$. Having computed $\pi_k$, the optimal power allocated to
jammers on link $\ell_k$ is given by the following expression:
\begin{equation}
\label{eq:pj}
   P_J^{(k)} = \frac{1}{\pi} \cdot x_k \sum_{\ell_i \in \Pi} x_i - y_k
        \eqend
\end{equation}
Using the above expression for $P_J^{(k)}$, the cost of link $\ell_k
\in \Pi$ is expressed as
\begin{equation}
\label{eq:lcost}
    \mc{C}(\ell_k)
        = \big( (\gamma_S k_{\rho}) \cdot d_{S_k, D_k}^{\alpha} - y_k \big)
            + \frac{1}{\pi} \big( x_k \sum_{\ell_i \in \Pi} x_i \big)
                \eqend
\end{equation}
Consequently, the cost of secure route $\Pi$ is given by:
\begin{equation}
\label{eq:pcost}
    \mc{C}(\Pi)
        = \sum_{\ell_k \in \Pi} \big( (\gamma_S k_{\rho}) \cdot d_{S_k, D_k}^{\alpha} - y_k \big)
            + \frac{1}{\pi} \big( \sum_{\ell_k \in \Pi} x_k \big)^2
            \eqend
\normalsize
\end{equation}
To this end, for a given route $\Pi$ between the source and
destination, the optimal cost of $\Pi$ subject to the end-to-end
eavesdropping constraint $\pi$ is given by~\eqref{eq:pcost}. The
optimal cost is achieved by allocating $P_S^{(k)}$ and $P_J^{(k)}$ to
each link $\ell_k \in \Pi$ using~\eqref{eq:ps} and \eqref{eq:pj},
respectively. Such a power allocation scheme would result in minimum
cost, while guaranteeing that the eavesdropping constraint would be
satisfied. Thus, SMER is reduced to finding a path, among all
possible paths between the source and destination, that minimizes the
optimal path cost~\eqref{eq:pcost}. The following proposition
formally states this result.
\begin{propose}
SMER with end-to-end eavesdropping and goodput constraint $\pi$ and
$\lambda$, respectively, is equivalent to finding a path that
minimizes the optimal path cost $\mc{C}(\Pi)$ as given
by~\eqref{eq:pcost}.
\end{propose}

\subsection{Optimal Path Cost Structure}
Define $\mc{C}_1(\ell_k)$ and $\mc{C}_2(\ell_k)$ as follows:
\begin{equation}
\begin{split}
    \mc{C}_1(\ell_k) &= (\gamma_S k_{\rho}) \cdot d_{S_k, D_k}^{\alpha} - y_k,\\
    \mc{C}_2(\ell_k) &= \frac{1}{\sqrt{\pi}} \cdot \sum_{\ell_k \in \Pi} x_k \eqend
\end{split}
\end{equation}
Then the optimal path cost~\eqref{eq:pcost} can be expressed as
\begin{equation}
\label{eq:pcost2}
    \mc{C}(\Pi)
        = \sum_{\ell_k \in \Pi} \mc{C}_1(\ell_k)
          + \Big( \sum_{\ell_k \in \Pi} \mc{C}_2(\ell_k) \Big)^2
            \eqend
\end{equation}

It is important to note that, while the $\mc{C}_1(\ell_k)$'s may
assume negative values, the path cost structure in (\ref{eq:pcost2})
is monotonous in the number of links, \ie\ if a path $\hat{\Pi}$ is a
subset of a path $\Pi$, then $\mc{C} (\hat{\Pi}) < \mc{C} (\Pi)$.
This is because $\pi <1$, and it can be shown that
$\big( \sum_{\ell_k \in \Pi} x_k \big)^2 > \sum_{\ell_k \in \Pi} y_k$.
Consequently, (\ref{eq:pcost2}) is minimized by a {\em simple} path.

\section{Secure Minimum Energy Routing}
\label{sec:algorithms}

In this section, we investigate the secure minimum energy routing
problem, where the cost of a path is given by~\eqref{eq:pcost}. We
begin by establishing that it is NP-hard. Then, by exploiting the
structure of the optimal solution, we employ dynamic programming to
obtain a pseudo-polynomial time algorithm that provides an exact
solution. This means that the problem is weakly NP-hard~\cite{garey},
thus fully polynomial time approximate schemes are possible.
Accordingly, we conclude the section by presenting a fully polynomial
time $\epsilon$-approximation algorithm for the problem,
which takes an approximation parameter $\epsilon > 0$ and after
running for time polynomial in the size of the network and in
$1/\epsilon$, it returns a path whose cost is at most $(1+\epsilon)$
times more than the optimal value.

\subsection{Computational Complexity}
We first show that our routing problem is NP-hard
via a reduction from the partition problem.

\begin{theorem}
Problem SMER is NP-hard.
\end{theorem}

\begin{proof}
We describe a polynomial time reduction of the Partition
problem~\cite{garey} to SMER. Given a set of integers $\mc{S}=
\set{k_1, k_2, \ldots, k_n}$, with $\sum_{i=1}^{n} k_i = 2\cdot K$,
the Partition problem is to decide whether there is a subset
$\mc{S'}$ of $\mc{S}$ such that $\sum_{i\in \mc{S'}} k_i = K$.

Given an instance $\mc{S}= \set{k_1, k_2, \ldots, k_n}$ of the
Partition problem, with $\sum_{i=1}^{n} k_i  = 2\cdot K$, we
construct the following network. The set of nodes is identical to
$\mc{S}$. For $i=1$ to $n-1$, we interconnect node $k_i$ to node
$k_{i+1}$ with two links, as follows: an ``upper'' link
$\ell_i^{(u)}$, to which we assign $\mc{C}_1(\ell_i^{(u)})= 2\cdot K
\cdot k_i$ and $\mc{C}_2(\ell_i^{(u)})= 0$, and a ``lower'' link
$\ell_i^{(w)}$, to which we assign $\mc{C}_1(\ell_i^{(w)})= 0$ and
$\mc{C}_2(\ell_i^{(w)})= k_i$.

\begin{lemma}
The answer to the Partition problem is affirmative {\em iff} the
solution to SMER in the constructed network, \ie\ the minimum value
of $(\ref{eq:pcost2})$ of a path between nodes $k_1$ and $k_n$,
equals $3 \cdot K^2$.
\end{lemma}

\begin{proof}
A path $\Pi$ between nodes $k_1$ and $k_n$ consists of a (possibly
empty) set of ``upper'' links $\mc{U}$ and a (possibly  empty) set
of ``lower'' links $\mc{W}$. Let $\mc{S}_u$ and $\mc{S}_w$ be,
correspondingly, the sets of indices of the links in $\mc{U}$ and in
$\mc{W}$, \ie\ $i\in \mc{S}_u$ {\em iff} $\ell_i^{(u)}\in \mc{U}$
and $i\in S_w$ {\em iff} $\ell_i^{(w)}\in \mc{W}$. Clearly,
$\mc{S}_u \cup \mc{S}_w = \mc{S}$. The cost of the path, per
$(\ref{eq:pcost2})$, is given by:
\begin{equation}
\label{eq:nph1}
\begin{split}
    \mc{C}(\Pi)
        &= \sum_{\ell_i^{(u)}\in \mc{U}} \mc{C}_1(\ell_i^{(u)}) + \sum_{\ell_i^{(w)}\in \mc{W}} \mc{C}_1(\ell_i^{(w)})\\
          &\quad+ \Big( \sum_{\ell_i^{(u)}\in \mc{U}} \mc{C}_2(\ell_i^{(u)}) + \sum_{\ell_i^{(w)}\in \mc{W}} \mc{C}_2(\ell_i^{(w)}) \Big)^2, \\
    &= \sum_{i\in \mc{S}_u} (2\cdot K \cdot k_i) + \sum_{i\in \mc{S}_w} 0
          + \Big( \sum_{i\in \mc{S}_u} 0 + \sum_{i\in \mc{S}_w} k_i \Big)^2
          \eqend
\end{split}
\end{equation}
Consider first the case $\sum_{i\in \mc{S}_u} k_i \geq \sum_{i\in
\mc{S}_w} k_i$. Since $\sum_{i\in \mc{S}_u} k_i + \sum_{i\in
\mc{S}_w} k_i = 2\cdot K$, denote: $\sum_{i\in \mc{S}_u} k_i =  K +
\delta$, $\sum_{i\in \mc{S}_w} k_i =  K - \delta$, for some $\delta
\geq 0$. Then, from (\ref{eq:nph1}), we have:
\begin{equation*}
\label{eq:nph2}
\begin{split}
    \mc{C}(\Pi)
        = 2\cdot K \cdot (K+\delta) + \Big( K-\delta \Big)^2 = 3 \cdot K^2 + \delta^2
    \eqend
\end{split}
\end{equation*}
Consider now the case $\sum_{i\in \mc{S}_u} k_i < \sum_{i\in
\mc{S}_w} k_i$. It follows similarly that
\begin{equation*}
\label{eq:nph3}
    \mc{C}(\Pi) = 3 \cdot K^2 + \delta^2 \eqend
\end{equation*}
We conclude that the length of a path between nodes $k_1$ and $k_n$
is at least $3 \cdot K^2$, and, furthermore, that value is attained
{\em iff} the set $\mc{S}$ can be partitioned into two subsets
$\mc{S}_u$ and $\mc{S}_w$, such that $\sum_{i\in \mc{S}_u} k_i =
\sum_{i\in \mc{S}_w} k_i$, \ie\ {\em iff} there is a subset $\mc
{S'}= \mc{S}_u$ of $\mc{S}$ such that  $\sum_{i\in \mc{S'}} k_i =
K$, and the lemma follows.
\end{proof}
Since the Partition problem is NP-complete~\cite{garey}, the
\mbox{theorem} follows.
\end{proof}

\subsection{Pseudo-Polynomial Time Exact Algorithm}
First, scale the values of the $\mc{C}_2(\ell)$'s for any link $\ell$
in the network so that they are all integers.\fnote{\mbox{The value
of ``$1$'' is determined by the precision at which we compute
$\mc{C}_2(\ell)$'s}.} Let $B$ denote an upper-bound on the sum of the
$\mc{C}_2(\ell)$'s on any simple path. A trivial bound is given by
$B=(N-1)\cdot\mc{C}_2^{max}$, where $N$ is the number of nodes in the
network and $\mc{C}_2^{max}$ is the maximum value of $\mc{C}_2(\ell)$
among all network links.
In a network with $N$
nodes, $\mc{C}_2^{max}$ can be computed in $O(N^2)$ time via a
brute-force search.

Our algorithm, termed DP-SMER, is listed below. \mbox{DP-SMER} iterates
over all values of $\mc{C}_2(\ell)$, \ie\
$\mc{C}_2(\ell)=1,2,\ldots,B$, and for each value of
$\mc{C}_2(\ell)$, it minimizes $\sum \mc{C}_1(\ell)$. Upon return,
the algorithm returns the cost of the optimal path from source $s$
to destination $d$ along with the structure $\Pi$ that contains the
network nodes that form the path.

\begin{algorithm}
\caption{DP-SMER (source $s$, dest. $d$, network $\mc{N}$).}
\label{alg:dp-smer}
\renewcommand{\algorithmiccomment}[1]{/*~#1~*/}
\begin{algorithmic}
    \STATE \COMMENT{path cost from $s$ to itself is always $0$}
        \FOR {$b = 1 \to B$}
                \STATE $C_s(b) = 0$
            \ENDFOR
    \STATE \COMMENT{initial path cost from $s$ to any other node is infinite}
        \FORALL {$n_i \in \mc{N}$, $n_i \neq s$}
            \FOR {$b = 1 \to B$}
                        \STATE $C_i(b) = \infty$
            \ENDFOR
        \ENDFOR

    \FOR {$b = 1 \to B$}
        \STATE \COMMENT{all node pairs can form a link and be neighbors}
        \FORALL {$n_i \in \mc{N}$}
            \FORALL {$n_j \in \mc{N}$}
                \STATE \COMMENT{update path cost via the neighboring nodes}
            \IF {$b+\mc{C}_2(\ell_{ij}) \leq B$}
                \STATE $ t = C_{i}(b) + \mc{C}_1(\ell_{ij})$
                \IF {$t < C_{j}(b+\mc{C}_2(\ell_{ij}))$}
                            \STATE $\Pi_j(b+\mc{C}_2(\ell_{ij})) = i$ \COMMENT {set $n_j$'s parent to $n_i$}
                            \STATE $C_{j}(b+\mc{C}_2(\ell_{ij})) = t$ \COMMENT {update path cost}
                \ENDIF
            \ENDIF
            \ENDFOR
            \ENDFOR
        \ENDFOR

    \STATE \COMMENT{include the ``$b$'' component, \ie\ $\mc{C}_2$, in the path costs}
    \FOR {$b = 1 \to B$}
            \STATE $\hat{C}_d(b) = C_d(b) + b^2$
    \ENDFOR

    \STATE \COMMENT{choose the best value for reaching the destination}
    \STATE $\displaystyle b^* = \arg \min_{b} \hat{C}_d(b)$
        \STATE {\bf return} $[\hat{C}_d(b^*), \Pi(b^*)]$
\end{algorithmic}
\end{algorithm}

\begin{theorem}
\label{theo:DP} DP-SMER runs in time $O(N^2 \cdot B)$, where $N$ is
the number of nodes in the network. Upon completion, the algorithm
returns an optimal solution to Problem SMER.
\end{theorem}
\begin{proof}
The first claim follows by noting that the computational complexity
is dominated by an iteration on all values $1,2,\ldots,B$, and for
each such iteration, iterating on all pairs of nodes.

We turn to consider the second claim. First, it can be established,
by induction on the values of $b$, that, upon completion of the
$b$-th iteration of the main loop of the algorithm, for all nodes
$n_i$, $C_i(b)$ is the length of a shortest path with respect to the
metric of the $\mc{C}_1(\ell_{ij})$ values, among all paths between
the source $s$ and node $n_i$, whose length with respect to the
metric of the $\mc{C}_2(\ell_{ij})$ values is precisely $b$.\fnote{We
note that this shortest path may be non-simple, \ie\, include loops,
due to the potentially negative values of $\mc{C}_1(\ell_{ij})$'s;
nonetheless, it is a finite path, and, furthermore, the optimal path
returned by the last step of DP-SMER is guaranteed to be simple, due
to the monotonicity property explained at the end of Section
\ref{sec:pathcost}.} Furthermore, it is easy to verify that the
values of $\hat{C}_N (b)$, computed at the next step of the
algorithm, stand for the lengths of the above shortest paths with
respect to the metric considered by Problem SMER.

Now, let $\Pi^*$ be an optimal solution (\ie\ a path) to Problem
SMER, and denote by $C^*$ its length with respect to the metric
considered by SMER. Furthermore, denote $b^* =
\sum_{\ell_{ij}\in\Pi^*} \mc{C}_2(\ell_{ij})$. It is easy to verify
that $\Pi^*$ is a shortest path with respect to the metric of the
$\mc{C}_1(\ell_{ij})$ values, among all paths between the source $s$
and the destination $d$, whose length with respect to the metric of
the $\mc{C}_2(\ell_{ij})$ values is precisely $b^*$. Therefore, upon
completion of the above steps of the algorithm, we will have
$\hat{C}_N (b^*) = C^*$; moreover, since $\Pi^*$ is an optimal
solution to SMER, it must hold that $\hat{C}_N (b^*) \leq \hat{C}_N
(b)$ for all values of $b$. The theorem follows.
\end{proof}

\subsection{Fully Polynomial Time $\epsilon$-Approximation}
As in the previous section, we scale the values of the
$\mc{C}_2(\ell)$'s for any link $\ell$ in the network so that they
are all integers and denote by $B$ an upper-bound on the sum of the
$\mc{C}_2(\ell)$'s on any simple path.

The above pseudo-polynomial solution indicates that SMER is only
weakly NP-hard (see~\cite{garey}), which enables us to apply
efficient, $\epsilon$-optimal approximation schemes of polynomial
time complexity, similar to the case of the widely investigated
Restricted Shortest Path problem (RSP, see,
\eg~\cite{Lorenz99asimple} and references therein). The RSP problem
considers a network where each link has two metrics, say ``cost''
and ``delay'', and some ``bound'' on the end-to-end delay. Then, for
a given source-destination pair, the problem is to find a path of
minimum cost among those whose delay do not exceed the delay bound.
This weakly NP-hard problem admits efficient $\epsilon$-optimal
approximation schemes of polynomial
complexity,~\eg~\cite{Lorenz99asimple}.

We turn to specify our approximation scheme for Problem SMER by a
simple employment of any solution to the RSP problem.\fnote{Other
solutions, of reduced computational complexity, can be established,
yet their structure is somewhat more complex.} First, a technical
difficulty arises in applying RSP approximation schemes to Problem
SMER. Recall that while link costs as given by~\eqref{eq:lcost} are
non-negative, $\mc{C}_1(\ell)$ can be negative for some links
$\ell$. In RSP, specifically in the approximation scheme
of~\cite{Lorenz99asimple}, it is assumed that link costs are
non-negative. Nevertheless, we show that the original network with
possibly negative link weights can be safely transformed (\ie\
without affecting the identity of the solution) to an expanded
network with non-negative link weights,  by employing the following
pre-processing step:
\begin{algorithm}
\caption{Expand\_Network (source $s$, network $\mc{N}$).}
\begin{enumerate}
\item Add the source node $s$ to the expanded network.
\item For each node $u$ $(u \neq s)$ in the original network,
add $N-1$ replicas denoted by $u(1), u(2), \ldots, u(N-1)$ to the
expanded network.

\item For each link $\ell_{su}$ from node $s$ to node $u$ in the original network,
add a link from node $s$ to node $u(1)$ in the expanded network with
the same metrics as for the original link.

\item For each link $\ell_{uv}$ in the original network, where \mbox{$u \neq s, u \neq d,
v \neq s$}, and for each $h=1,\ldots, N-2$, add a link between node
$u(h)$ and node $v(h+1)$ in the expanded network with the same
metrics as for the original link.

\item For each link $\ell$ in the expanded network, add some (identical to all links)
bias $\delta \ge 0$
to each link cost $\mc{C}_1(\ell)$ so that the new link costs would
be non-negative.
\end{enumerate}
\end{algorithm}

The following lemmas establish the relation between the shortest
paths in the original network and the shortest paths in the expanded
network.
\begin{lemma}
\label{l:exp}
A path that is shortest w.r.t. the biased metric $(\mc{C}_1(\ell) +
\delta)$ among those that obey a bound on the $\sum \mc{C}_2(\ell)$
and have precisely $h$ hops, is also shortest w.r.t. the unbiased
metric $\mc{C}_1(\ell)$ among those that obey the same bound on
$\sum \mc{C}_2(\ell)$ and have precisely $h$ hops.
\end{lemma}
\begin{proof}
Suppose that this is not true. That is, there are paths $\Pi$ and
$\Pi'$, both obeying the bound on $\sum \mc{C}_2(\ell)$ and with $h$
hops, in such a way that $\Pi'$ is a shortest path with the bias yet
$\Pi$ is shorter without the bias. Therefore, $\sum_{\ell \in \Pi}
\mc{C}_1(\ell) < \sum_{\ell \in \Pi'} \mc{C}_1(\ell)$, yet
$\sum_{\ell \in \Pi} (\mc{C}_1(\ell) + \delta) \ge \sum_{\ell \in
\Pi'} (\mc{C}_1(\ell) + \delta)$. However, the second inequality can
be rewritten as:\linebreak $\sum_{\ell \in \Pi} \mc{C}_1(\ell) + h \cdot
\delta \ge \sum_{\ell \in \Pi'} \mc{C}_1(\ell) + h \cdot \delta$,
which contradicts the first inequality.
\end{proof}

\begin{lemma}
\label{l:hops} A shortest path from source $s$ to node $d(h)$ in the
expanded network has precisely $h$ hops.
\end{lemma}
\begin{proof}
The proof follows from the fact that the $i$-th hop on the shortest path
from $s$ to $d(h)$ has to go from some node $v(i-1)$ to some node
$u(i)$ (see the network expansion procedure).
\end{proof}

Thus, to compute an $\epsilon$-optimal solution to Problem SMER, for
every bound on $\sum \mc{C}_2(\ell)$, we find the shortest path with
$h=1,\ldots,N-1$ hops in the expanded network by repeatedly
employing an approximation solution to the RSP problem.
For a given
approximation value $\epsilon>0$, let \mbox{$\eta=\epsilon/3$}.
Furthermore, let $L$ be the smallest integer for which $\lceil
(1+\eta)^L \rceil \geq B$. Our algorithm, called $\epsilon$-SMER, is
listed below. In this algorithm, $\epsilon$-RSP refers to an
$\epsilon$-optimal approximation solution for the RSP problem.

\begin{algorithm}
\caption{$\epsilon$-SMER (error $\epsilon$, source $s$, dest. $d$,
net. $\mc{N}$).}
\renewcommand{\algorithmiccomment}[1]{/* #1 */}
\begin{algorithmic}
    \STATE $\mc{N}_x$ = Expand\_Network($s$, $\mc{N}$)

    \FORALL {$\ell \in \mc{N}_x$}
        \STATE $cost(\ell) = \mc{C}_1(\ell)$
        \STATE $delay(\ell) = \mc{C}_2(\ell)$
    \ENDFOR

    \FOR {$l = 1 \to L$}
        \STATE $delay\_bound = \lceil (1+\eta)^l \rceil$
        \STATE \COMMENT{compute the approximate $h$-hop path}
        \FOR {$h = 1 \to N-1$}
            \STATE $[C(l, h), \Pi(l, h)] =$ $\epsilon$-RSP$(\epsilon, s, d(h), \mc{N}_x)$
         \STATE \COMMENT{compute the actual cost as per SMER metric}
        \STATE $\hat{C}(l,h) = (C(l,h) - h \cdot \delta) + \lceil (1+\eta)^l \rceil^2$
        \ENDFOR
    \ENDFOR

    \STATE \COMMENT{choose the best $l$ and $h$ for reaching the destination}
    \STATE $\displaystyle (l^*,h^*) = \arg \min_{l,h} \hat{C}(l,h)$
        \STATE {\bf return} $[\hat{C}(l^*,h^*), \Pi(l^*,h^*)]$
\end{algorithmic}
\end{algorithm}

In the $\epsilon$-SMER algorithm, for each considered delay bound
$\lceil (1+\eta)^l \rceil$, $N-1$ instances of the approximation
solution to the RSP problem, for the same bound, are run on the
expanded network: in each instance $h$, we consider $s$ to be the
source and $d(h)$ to be the destination. Using Lemma~\ref{l:hops},
it is straightforward to verify that, in each instance $h$, the RSP
approximation obtains a solution that satisfies the required delay
bound with the restriction that the path has {\em precisely} $h$
hops (in both the expanded and the original network).

Therefore, per considered bound on the $\mc{C}_2(\ell)$ metric and
per possible number of hops up to $N-1$, we get an $\epsilon$-optimal
path with respect to the original metric $\mc{C}_1(\ell)$ (of
precisely that many hops). It follows from Lemmas~\ref{l:exp} and
~\ref{l:hops}, that, by comparing all solutions (for all considered
bounds on the $\mc{C}_2(\ell)$ metric and number of hops $h$), we
will find a shortest $\epsilon$-optimal path that corresponds to an
$\epsilon$-optimal solution to SMER. This is established next
through the following lemmas and theorem.

\begin{lemma}
\label{l:eps1} Let $\Pi^*$ be an optimal solution (path) to SMER.
Denote by $\mc{C}(\Pi^*)$ and $\mc{C}(\hat{\Pi})$, the costs, per
the SMER metric, of the optimal solution and of the solution
obtained by $\epsilon$-SMER, correspondingly. Then:
\begin{equation}
\label{eq:eps1}
    \mc{C}(\hat{\Pi}) \leq (1+\epsilon) \cdot \mc{C}(\Pi^*)
    \eqend
\end{equation}
\end{lemma}
\begin{proof}
Let $\bar{l}$ be the smallest integer such that
\begin{equation}
\label{eq:eps2}
    \Big( \sum_{\ell \in \Pi^*} \mc{C}_2(\ell)
    \Big)^2 \leq \Big( \lceil {(1+\eta)^{\bar{l}}} \rceil \Big)^2
    \eqend
\end{equation}
Note that this implies that:
\begin{equation}
\label{eq:eps3}
    (1+\eta)^2 \Big( \sum_{\ell \in \Pi^*}
    \mc{C}_2(\ell) \Big)^2 \geq \Big( \lceil (1+\eta)^{\bar{l}} \rceil
    \Big)^2
    \eqend
 \end{equation}
Let $\bar{h}$ be the number of hops of $\Pi^*$. By construction,
$\Pi (\bar{l},\bar{h})$ is an $\epsilon$-optimal approximation for
RSP, for ``costs'' $\mc{C}_1(\ell)$, ``delays'' $\mc{C}_2(\ell)$,
``delay bound'' $\lceil (1+\eta)^{\bar{l}} \rceil$ and precisely
$\bar{h}$ hops. Moreover, by (\ref{eq:eps2}), the path $\Pi^*$ obeys
this bound. Therefore:
\begin{equation}
    \sum_{\ell \in \Pi (\bar{l},\bar{h})} \mc{C}_1(\ell) + \bar{h} \cdot \delta \leq
    (1+\epsilon) \sum_{\ell \in \Pi^*} \mc{C}_1(\ell) + \bar{h} \cdot \delta,
 \end{equation}
or, equivalently,
\begin{equation}
\label{eq:eps4}
    \sum_{\ell \in \Pi (\bar{l},\bar{h})} \mc{C}_1(\ell)
    \leq (1+\epsilon) \sum_{\ell \in \Pi^*} \mc{C}_1(\ell)
    \eqend
 \end{equation}
Since $\Pi (\bar{l},\bar{h})$ obeys the ``delay bound'' $\lceil (1+\eta)^{\bar{l}}
\rceil$, we have:
\begin{equation}
\label{eq:eps5}
    \Big( \sum_{\ell \in \Pi (\bar{l},\bar{h})}
    \mc{C}_2(\ell) \Big)^2 \leq \Big( \lceil (1+\eta)^{\bar{l}} \rceil \Big)^2
    \eqend
 \end{equation}
Combining (\ref{eq:eps3}), (\ref{eq:eps4}) and (\ref{eq:eps5}), we
have:
\begin{equation}
\label{eq:eps6}
\begin{split}
    &\mc{C}(\hat{\Pi}) \leq \mc{C}( \Pi (\bar{l},\bar{h}))\\
    &\quad\leq (1+\epsilon) \sum_{\ell \in \Pi^*} \mc{C}_1(\ell) +
    (1+\eta)^2 \Big( \sum_{\ell \in \Pi^*} \mc{C}_2(\ell) \Big)^2,
\end{split}
\end{equation}
where the first transition is due to the way that $\hat{\Pi}$ is
chosen. Since $\eta={\epsilon\over 3}$, for small values of
$\epsilon$ (precisely, $\epsilon < 3$), (\ref{eq:eps6}) implies:
\begin{equation}
\label{eq:eps7} \mc{C}(\hat{\Pi}) \leq (1+\epsilon) \cdot
\mc{C}(\Pi^*),
\end{equation}
as required.
\end{proof}

\begin{lemma}
\label{l:eps2} The computational complexity of $\epsilon$-SMER is
\mbox{$O(A\cdot {1\over \epsilon} \cdot \log (B) \cdot N^3)$}, where
$O(A)$ is the computational complexity of the employed approximation
scheme for RSP.
\end{lemma}
\begin{proof}
Let $M$ be the number of links in the original network. Each time we
employ the RSP approximation scheme, we would incur a computational
complexity of $O(A)$, where $A$ corresponds to a network with $N$
nodes and $M$ links.

For each value of $l=1\ldots,L$, we call the RSP approximation as
follows: once for a network with $N$ nodes and $O(N)$ links (\ie\
for the network that contains $s$ and all the $u(1)$'s), once for a
network with roughly $2N$ nodes and $M$ links (\ie\ for the network
that contains, in addition to the above, all the $u(2)$'s and links
of the form $(u(1),v(2))$, once for a network with roughly $3N$
nodes and $2M$ links (\ie\ for the network that contains, in
addition to the above, all the $u(3)$'s and links of the form
$(u(2),v(3))$, and so on up to, once (the $(N-1)$-th time) for a
network with roughly $(N-1)N$ nodes and $(N-2)M$ links. The above
$N-1$ instances (more precisely, all but the first, which can be
neglected due to smaller complexity) aggregate to:

\begin{equation}
\begin{split}
    &O\big(A
      \cdot (2\cdot1+3\cdot2+\cdots+N\cdot(N-1)) \big) \\
    &\quad= O\big( A \cdot \sum_{i=1}^{N-1} i(i+1) \big)
    = O\big( A \cdot N^3 \big)
    \eqend
\end{split}
\end{equation}

The proof follows by noting that $L=O({1 \over \epsilon} \cdot \log
(B))$.
\end{proof}

\begin{theorem}
$\epsilon$-SMER is an $\epsilon$-optimal approximation scheme of
polynomial complexity. In particular, when employing the
approximation solution of~\cite{Lorenz99asimple} to the RSP problem,
$\epsilon$-SMER runs in $O(N^6 \cdot (\log \log N + {1 \over
\epsilon}) \cdot {1\over \epsilon} \cdot \log (B))$ time.
\end{theorem}
\begin{proof}
The RSP scheme of~\cite{Lorenz99asimple} has computational
complexity of $O\big( (N \cdot M \cdot (\log \log N+1/\epsilon)
\big)$ for $N$ nodes and $M$ links. Depending on the limit on the
transmission power at each node, in worst-case we have $M=O(N^2)$,
\ie\ all nodes may be neighbors\fnote{Note that, typically, the network is sparse, \ie\ $M \ll N^2$, hence
the dependency on $N$ is more like $N^5$.}.
The proof then follows from
Lemmas~\ref{l:eps1}~and~\ref{l:eps2}.
\end{proof}

More efficient versions of $\epsilon$-SMER should be possible,
yet our goal has been to show that fully polynomial time $\epsilon$-approximation schemes
(FPTAS) exist for the NP-hard problem SMER.

\subsection{Distributed Implementation}
While it is not discussed in this paper, our routing algorithms can
be implemented in a distributed manner following standard techniques
of distance-vector routing. Note that the power allocation at the
physical layer is a local operation performed by the transmitting
node of each link based on the information from the routing algorithm
and topological information (collected, for instance, through
neighbor discovery before running the routing algorithm).

\section{Simulation Results}
\label{sec:simulation}

\subsection{Simulation Environment}
We have implemented our routing algorithms in a custom-built
simulator to study their performance in a variety of network
scenarios. We simulate a wireless network, in which nodes are
distributed uniformly at random in a square of area $5 \times 5$ with
node density $\sigma=3$. We also place a number of eavesdroppers in
the network with density $\sigma_E$, as described later. We consider
one eavesdropper per link. We keep the number of eavesdroppers
considerably less than that of the legitimate nodes in order to be
able to establish secure routes as we put a limit on the maximum
transmission power of each node. Every node has a maximum
transmission power that is set in such a way that the resulting
network becomes connected (the absolute value of the maximum power
does not affect the results). We choose two nodes $s$ and $d$ located
at the lower left and the upper right corners of the network,
respectively, and find paths from $s$ to $d$. We then compute the
total amount of energy consumed on each path using different routing
algorithms. The performance metric ``energy savings'' refers to the
percentage difference between total energy used by different
algorithms with respect to the benchmark. For simulation purposes, we
set $\pi=0.1$, $\sigma_E=1$, $N_0=1$, $\gamma_D = 0.8$, and $\gamma_E
= 0.6$, unless otherwise specified. The numbers reported are obtained
by averaging over $10$ simulation runs with different seeds.

\subsection{Simulated Algorithms}
In addition to
DP-SMER and \mbox{$\epsilon$-SMER}, we have also implemented a
security-agnostic algorithm based on minimum energy routing as a
benchmark to measure energy savings achieved by our algorithms. The
benchmark algorithm, called {\em security-agnostic shortest path
routing (SASP)}, is described below. Note that some of the
optimizations described in Sections~\ref{sec:linkcost}
and~\ref{sec:pathcost} have been incorporated in SASP, making it a
considerably efficient benchmark (see Subsection~\ref{sub:alloc}).

\begin{algorithm}
\caption{SASP (source $s$, dest. $d$, network $\mc{N}$).}
\begin{enumerate}
\item Find a shortest path in terms of transmission power between
$s$ and $d$ ignoring eavesdroppers. The standard Dijkstra's
algorithm can be used for this purpose.

\item Use~\eqref{eq:pi} to allocate an optimal eavesdropping probability
to each link of the computed path.

\item Use~\eqref{eq:pj} to allocate sufficient power to jammers on each link
with respect to the allocated eavesdropping probabilities in step
(2).
\end{enumerate}
\end{algorithm}

\subsection{Results and Discussion}

\spar{Effect of Eavesdropper Location on Link Cost}
For a fixed link between two nodes, the source transmission power is
also fixed as obtained in~\eqref{eq:ps}. Thus, the cost of the link
depends only on the jamming power which is a function of the
eavesdropper location as given by~\eqref{eq:pj1}.
Fig.~\ref{fig:onelink} shows the cost of establishing a secure link
between source $S$ (placed at the center) and destination $D$ for
different eavesdropper locations and $\pi = 0.001$. In the figure,
the color intensity at each point is proportional to the amount of
energy required to establish the link if the eavesdropper is placed
at that point. Clearly, by some maneuvering around an eavesdropper,
a significant reduction in energy cost can be achieved as the
eavesdropper becomes almost ineffective in some locations. This is
the main idea behind this work.

\begin{figure}[ht]
    \centering
    \subfloat[Path-loss exponent $\alpha=2$.]{
        \includegraphics[width=7cm]{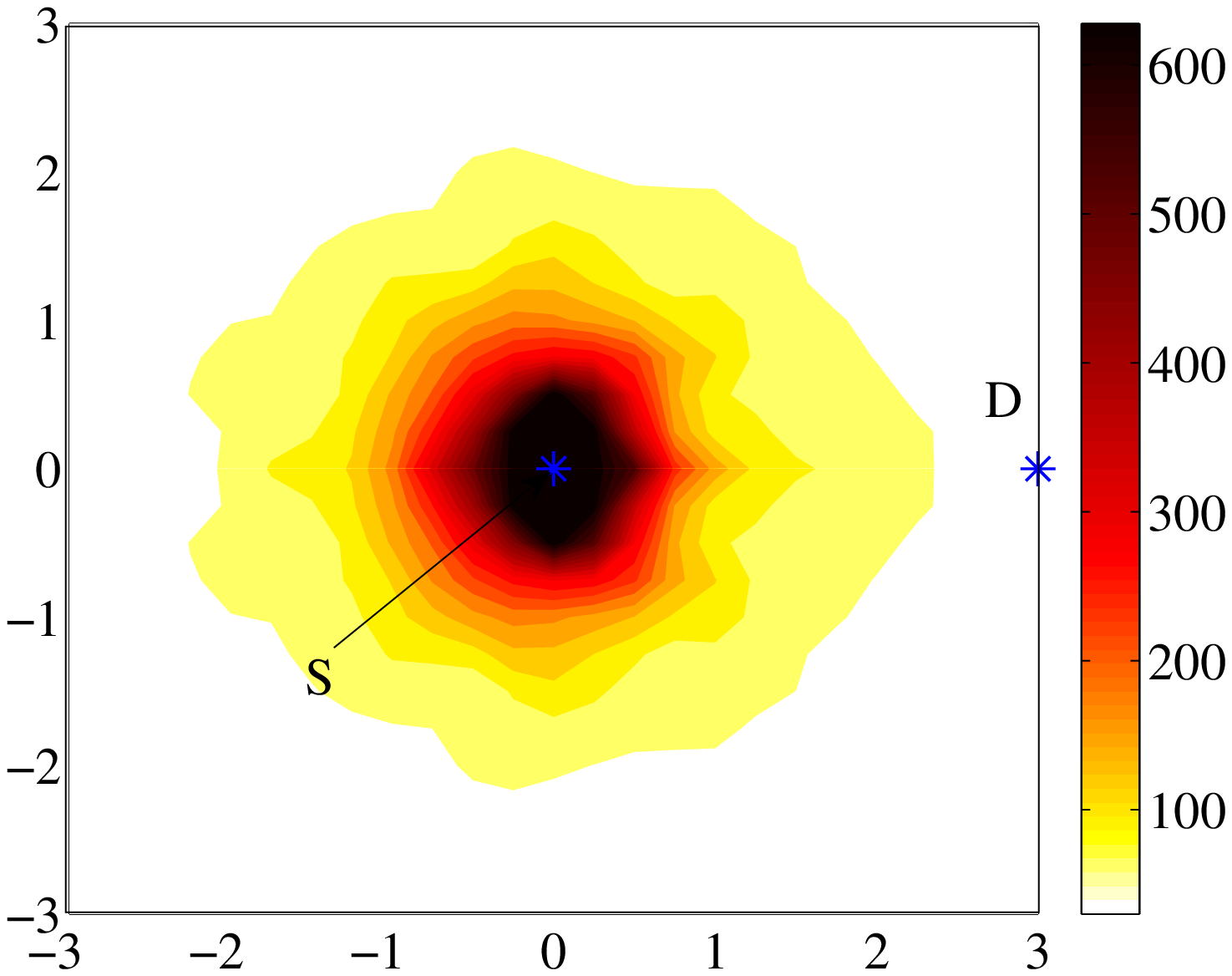}
        \label{fig:onelinka2}
    }
    \hspace{1mm}
    \subfloat[Path-loss exponent $\alpha=4$.]{
        \includegraphics[width=7cm]{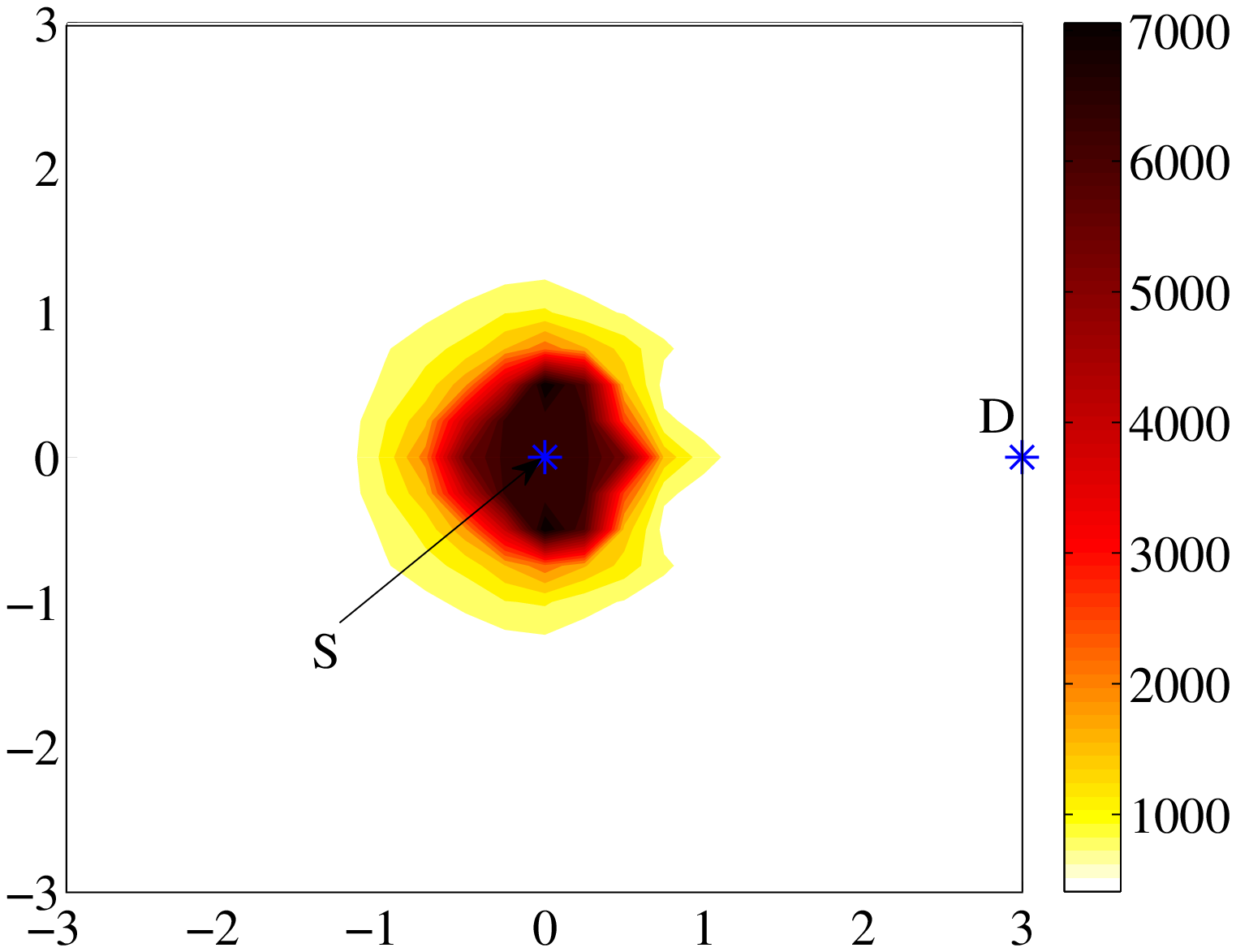}
        \label{fig:onelinka4}
    }
  \caption{Effect of eavesdropper location on link cost.}
  \label{fig:onelink}
\end{figure}

\spar{Effect of Optimal Secrecy Allocation on Path Cost}
\label{sub:alloc} For a fixed path subject to an end-to-end secrecy
requirement $\pi$, the optimal eavesdropping probability assigned to
each link of the path is given by~\eqref{eq:pi}, which in turn
determines the optimal jamming power allocated to each link of the
path using~\eqref{eq:pj}. Specifically, this is how power allocation
is performed in SASP in order to minimize power consumption.
Alternatively, a simple heuristic is to divide $\pi$ equally across
the links. That is, if the path contains $h$ links, then each link
$\ell_k$ is allocated sufficient jamming power to satisfy the
eavesdropping probability $\pi_k = \pi/h$. In
Fig.~\ref{fig:allocation}, we have depicted energy savings that can
be achieved ``solely'' by optimal secrecy allocation compared to
equal allocation for a fixed path that is computed by SASP.
Interestingly, as the number of eavesdroppers increases or the
signal propagation becomes more restricted, optimal secrecy
allocation becomes even more important, achieving energy savings of
up to $72\%$ ($47\%$) for $\alpha=4$ ($\alpha=2$) in the simulated
network.

\begin{figure}[ht]
    \centering
    \subfloat[Effect of eavesdropper density.]{
        \includegraphics[width=8cm]{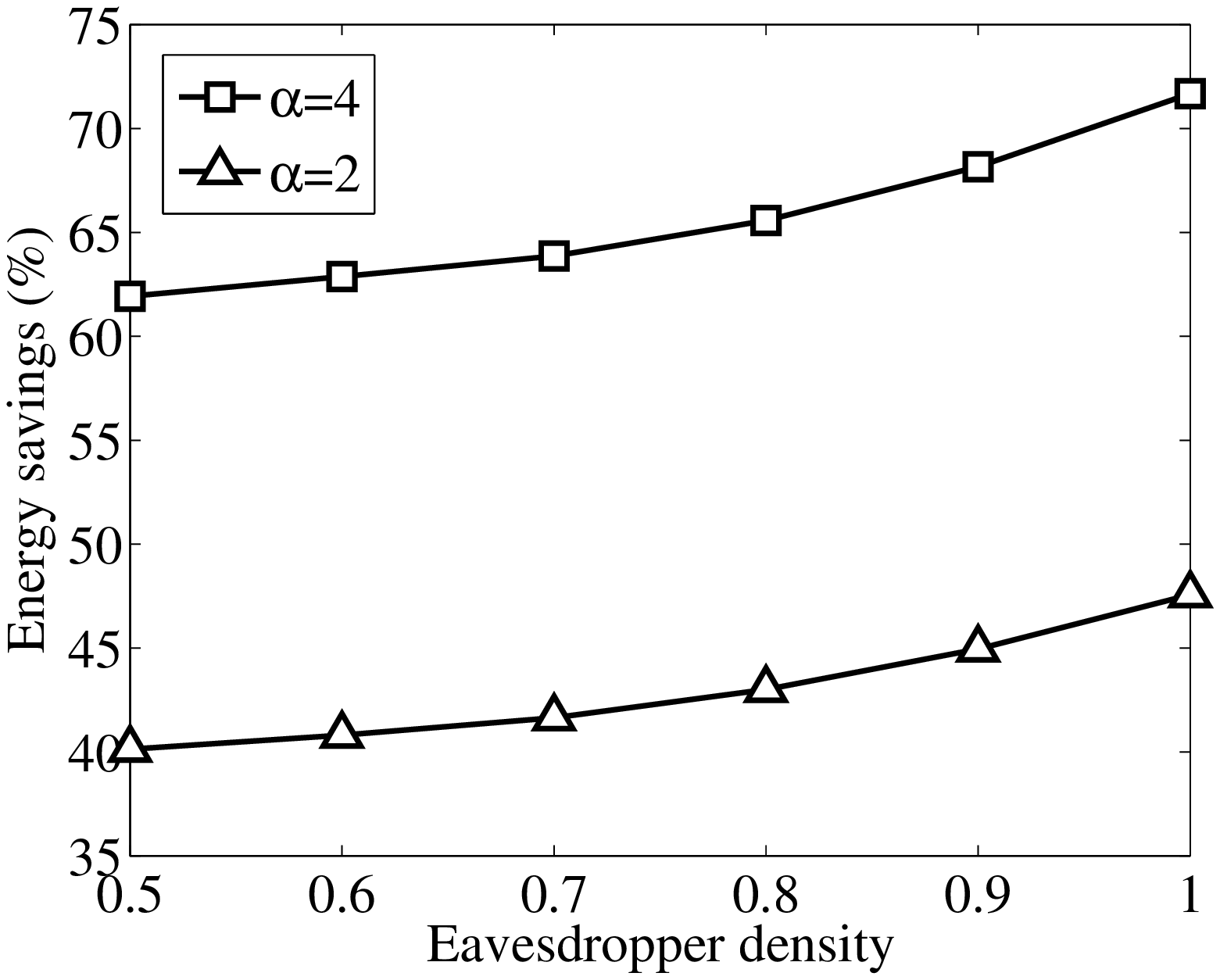}
        \label{fig:optdens}
    }
    \subfloat[Effect of eavesdropping prob.]{
        \includegraphics[width=8cm]{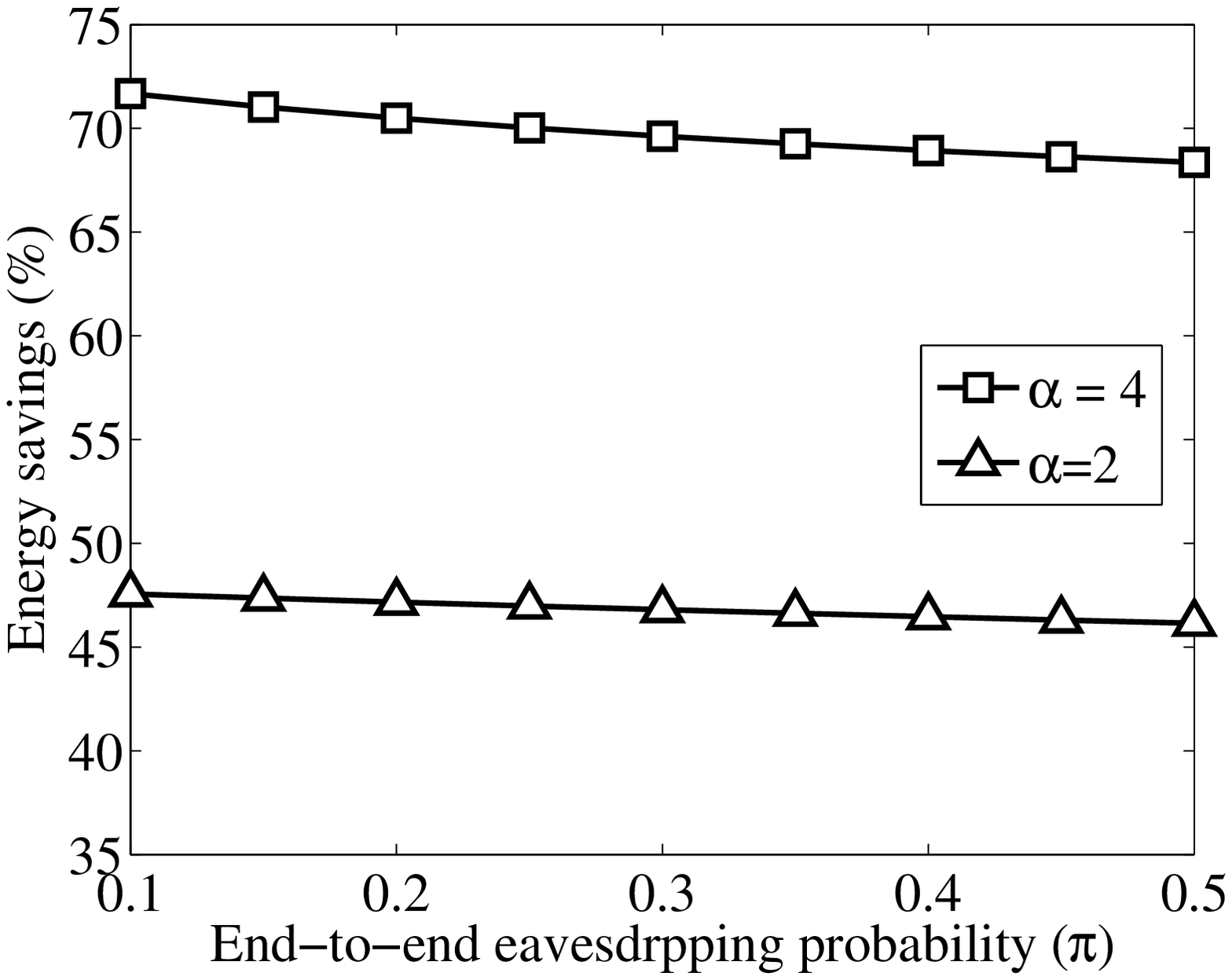}
        \label{fig:optprob}
    }
  \caption{Energy savings achieved by optimal secrecy allocation.}
  \label{fig:allocation}
\end{figure}

\spar{Non-uniform Eavesdropper Placement}
To gain more insight about the behavior of different routing
algorithms, in this experiment, rather than randomly distributing
eavesdroppers in the network, we strategically place them close to
the line that connects the source and destination. Ideally, SMER and
$\epsilon$-SMER should avoid the shortest path that crosses the
network diagonally. This is indeed the behavior observed in the
simulations as depicted in Fig.~\ref{fig:nonuniform} (`\scalebox{0.6}{$\bigstar$}'
denotes an eavesdropper). As expected, SASP blasts right through the
eavesdroppers, while SMER, $0.1$-SMER and $1.0$-SMER route around
them resulting in $88\%$, $86\%$ and $85\%$ energy savings,
respectively.
\begin{figure}[ht]
  \centering
  \includegraphics[width=8cm]{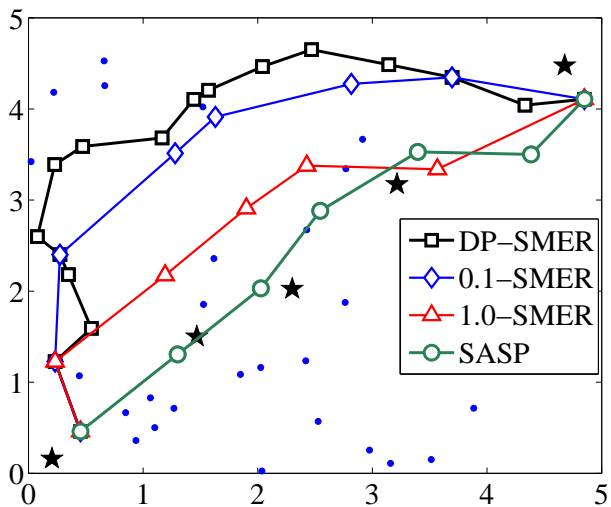}
  \caption{Snapshot of paths computed by different algorithms.}
  \label{fig:nonuniform}
\end{figure}

\spar{Uniform Eavesdropper Placement} In this experiment,
eavesdroppers are placed in the network uniformly at random. As seen
in Fig.~\ref{fig:uniform}, our algorithms consistently outperform
SASP for a wide range of eavesdropper densities and eavesdropping
probabilities. In particular, energy savings of up to $99\%$ and
$98\%$ (for $\alpha=4$) can be achieved by SMER and $0.1$-SMER,
respectively.

\begin{figure}[ht]
    \centering
    \subfloat[Effect of eavesdropper density.]{
        \includegraphics[width=8cm]{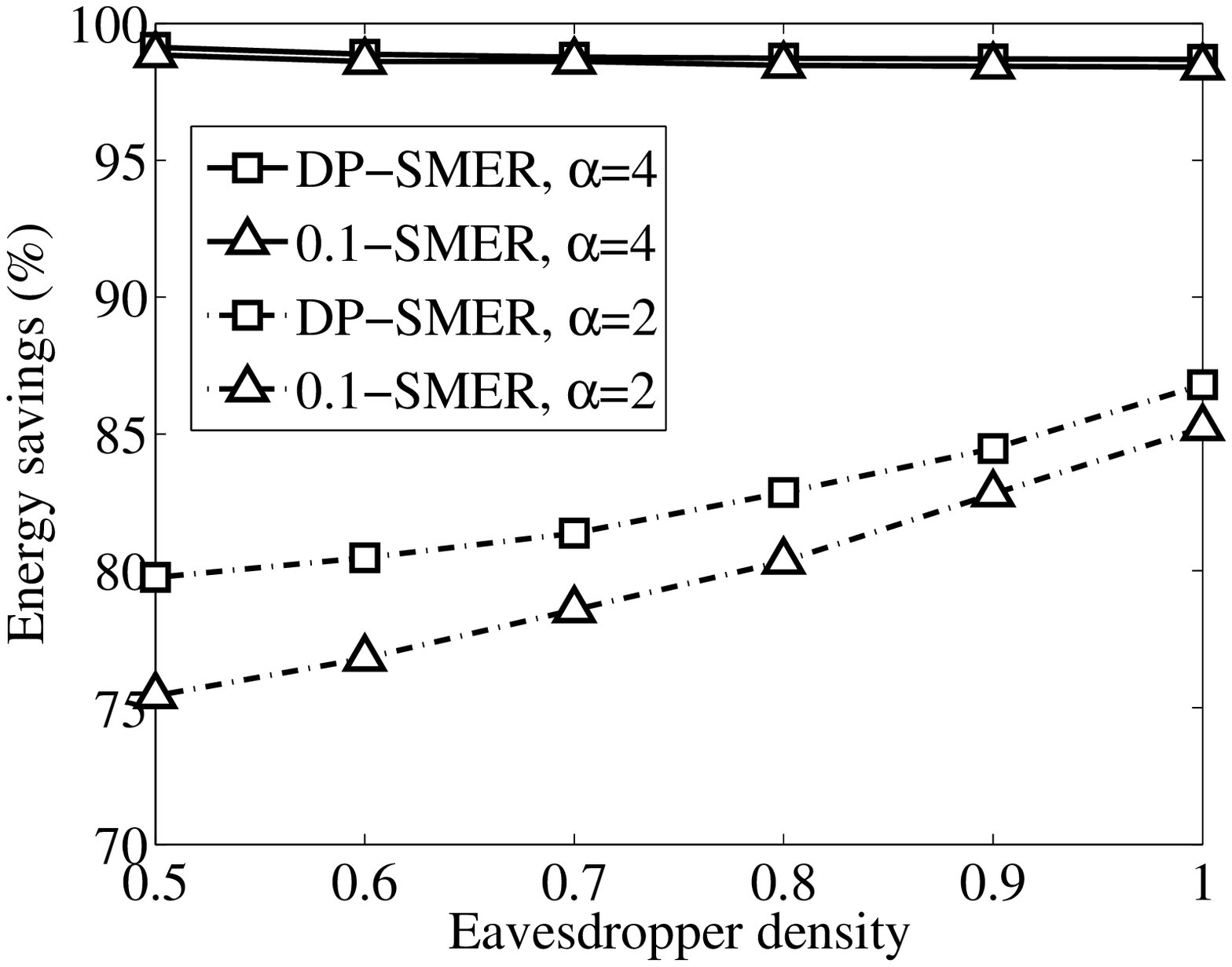}
        \label{fig:uniform1}
    }
    \subfloat[Effect of eavesdropping prob.]{
        \includegraphics[width=8cm]{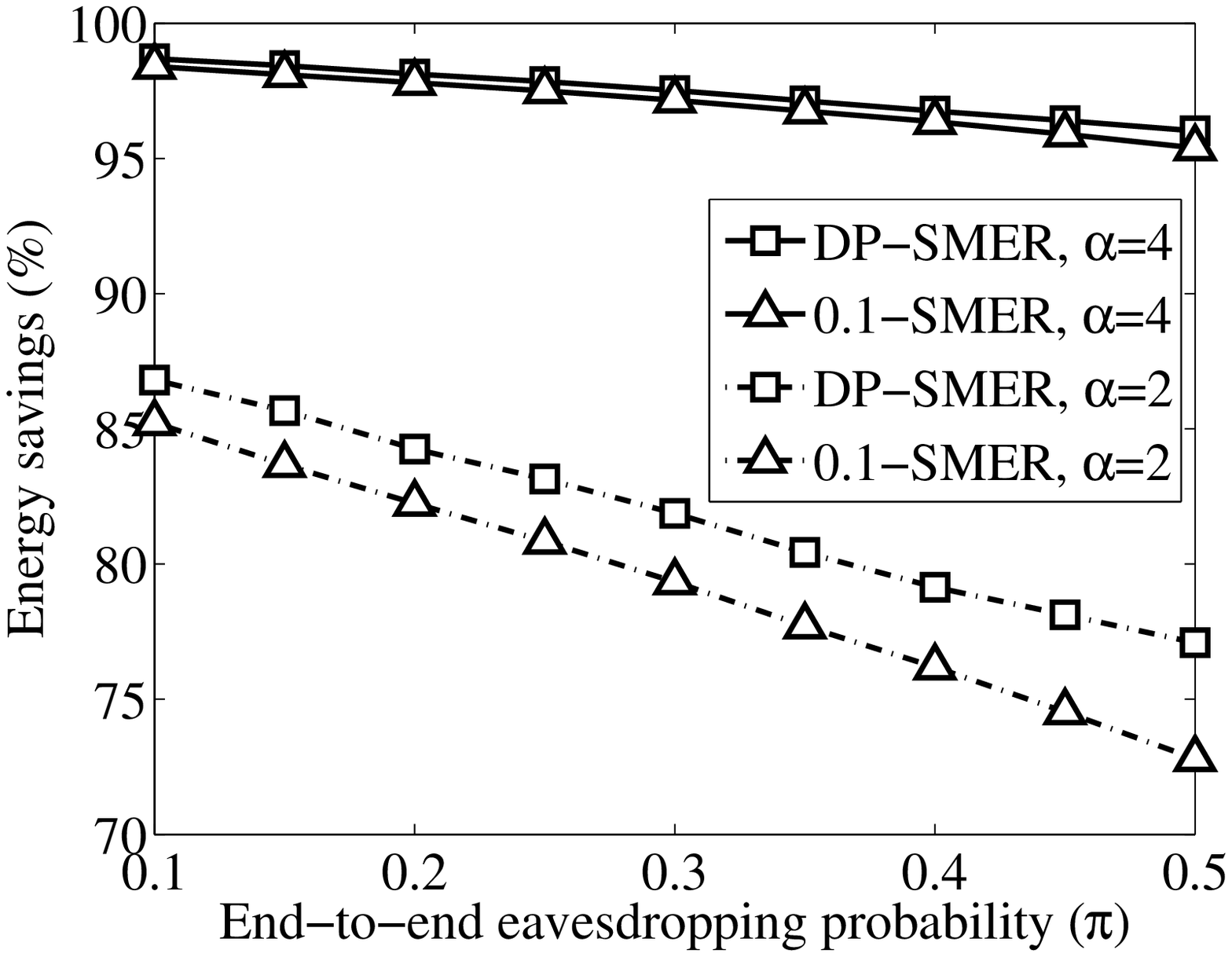}
        \label{fig:uniform2}
    }
  \caption{Energy savings with uniform eavesdropper placement.}
  \label{fig:uniform}
\end{figure}

\spar{Effect of Network Size}
Fig.~\ref{fig:size} shows the energy savings achieved by different
algorithms in networks with varying sizes. The ``network dimension''
refers to the length of one side of the square area that contains
the network nodes. As observed from the figure, the energy saving
is an increasing function of the network size. Interestingly, as the
network size increases, the effect of the propagation environment
diminishes in such a way that energy savings for $\alpha=2$ and
$\alpha=4$ converge to the same numbers as opposed to the previous
scenarios. As the network size increases so does the average length
of the path (in terms of the number of hops) between the source and
destination nodes. Those paths that are longer provide more
opportunities for energy savings on each link of the path resulting
in increased overall energy savings. This effect works in favor of
$\alpha=2$ as well as $\alpha=4$. However, given the high values of
energy savings for $\alpha=4$ (due to longer paths compared to
$\alpha=2$), the effect of longer paths is more prominent for
$\alpha=2$.
\begin{figure}[ht]
    \centering
    \subfloat[Non-uniform eavesdrop. placement.]{
        \includegraphics[width=8cm]{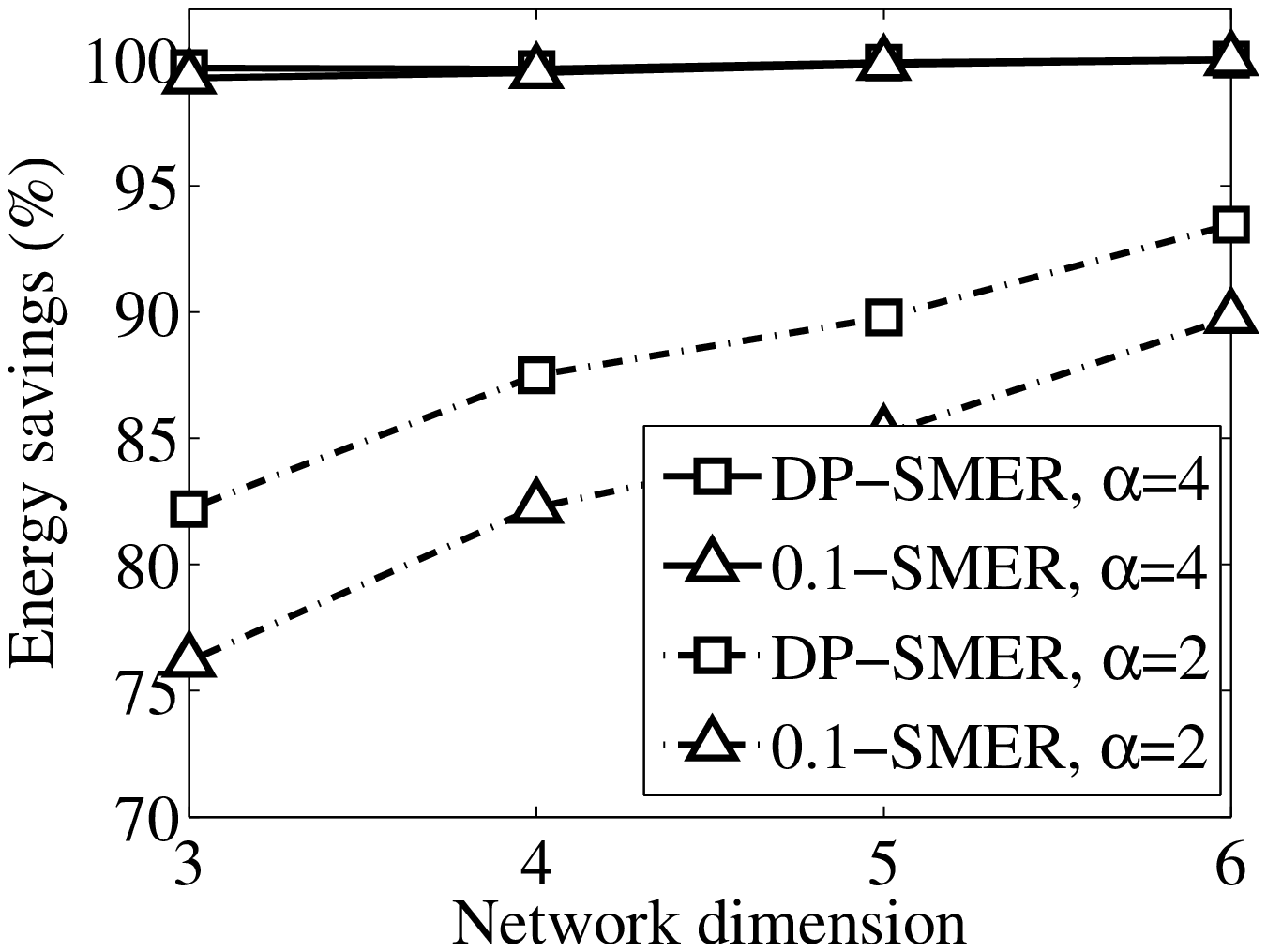}
        \label{fig:size_diag}
    }
    \subfloat[Uniform eavesdrop. placement.]{
        \includegraphics[width=8cm]{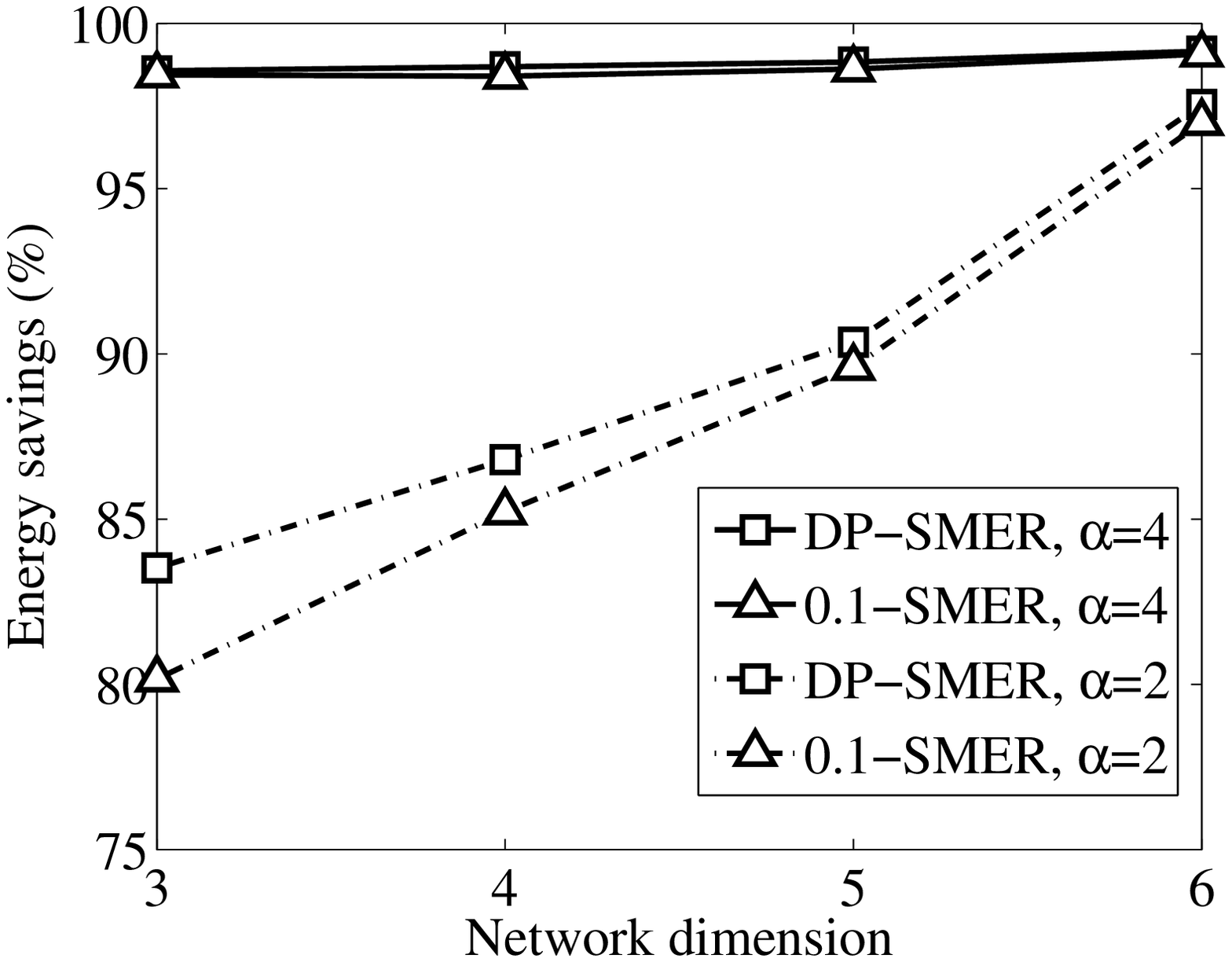}
        \label{fig:size_uniform}
    }
  \caption{Effect of network size on energy savings.}
  \label{fig:size}
\end{figure}

\spar{Effect of Jamming Set} The cardinality of the jamming set
affects the power allocation to jammers. In this experiment, we
change the number of jammers that participate in secure transmissions
on each link and compute the energy savings achieved by different
algorithms. Figs.~\ref{fig:jammers}\subref{fig:jammers_diag}
and~\ref{fig:jammers}\subref{fig:jammers_uniform}, respectively, show
the energy savings achieved for non-unform and uniform placement of
eavesdroppers. Interestingly, in these scenarios, a small number of
jammers, namely $2$, is sufficient to obtain most of the benefits of
cooperative jamming, which should greatly simplify any practical
implementation.

\begin{figure}[ht]
    \centering
    \subfloat[Non-uniform eavesdrop. placement.]{
        \includegraphics[width=8cm]{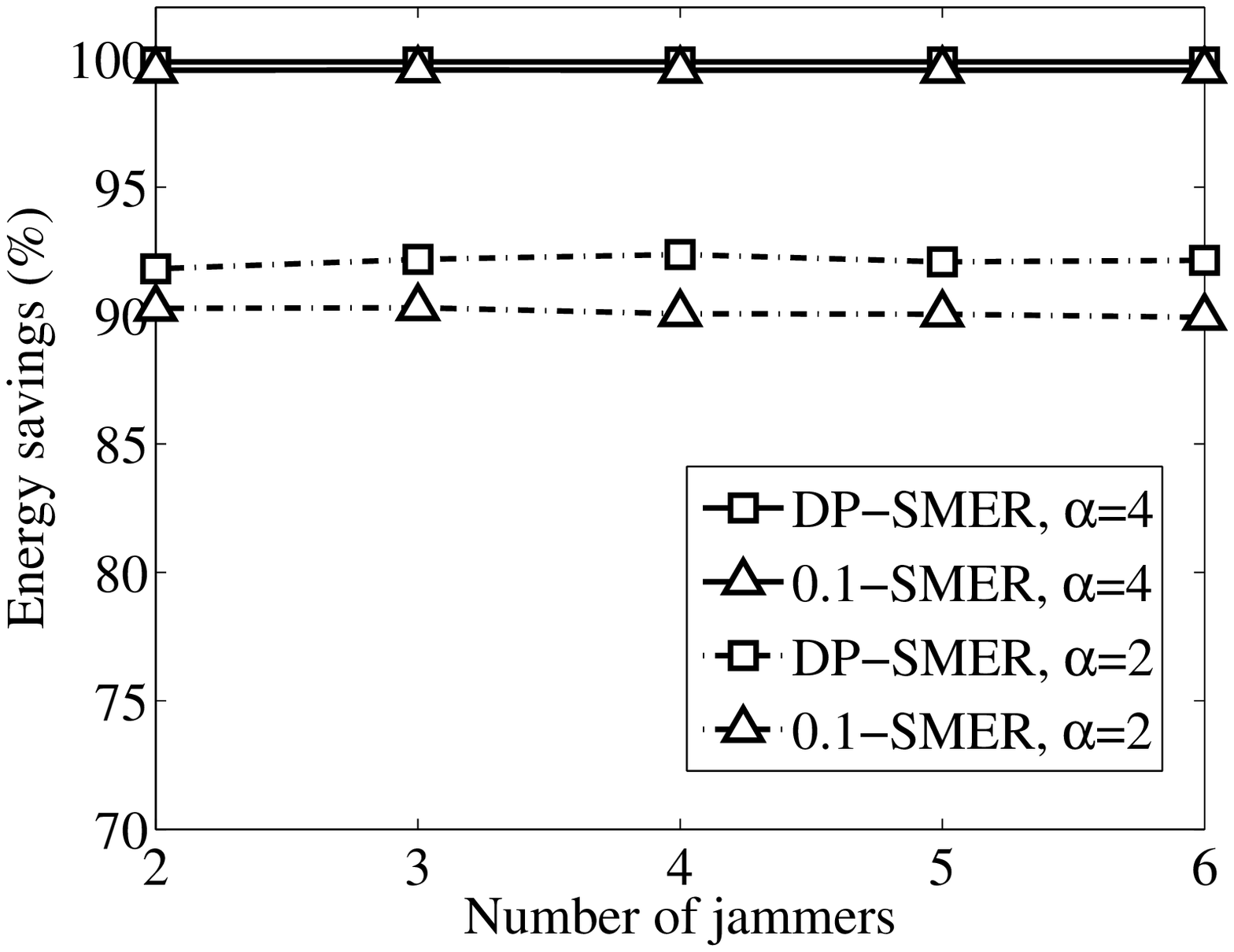}
        \label{fig:jammers_diag}
    }
    \subfloat[Uniform eavesdrop. placement.]{
        \includegraphics[width=8cm]{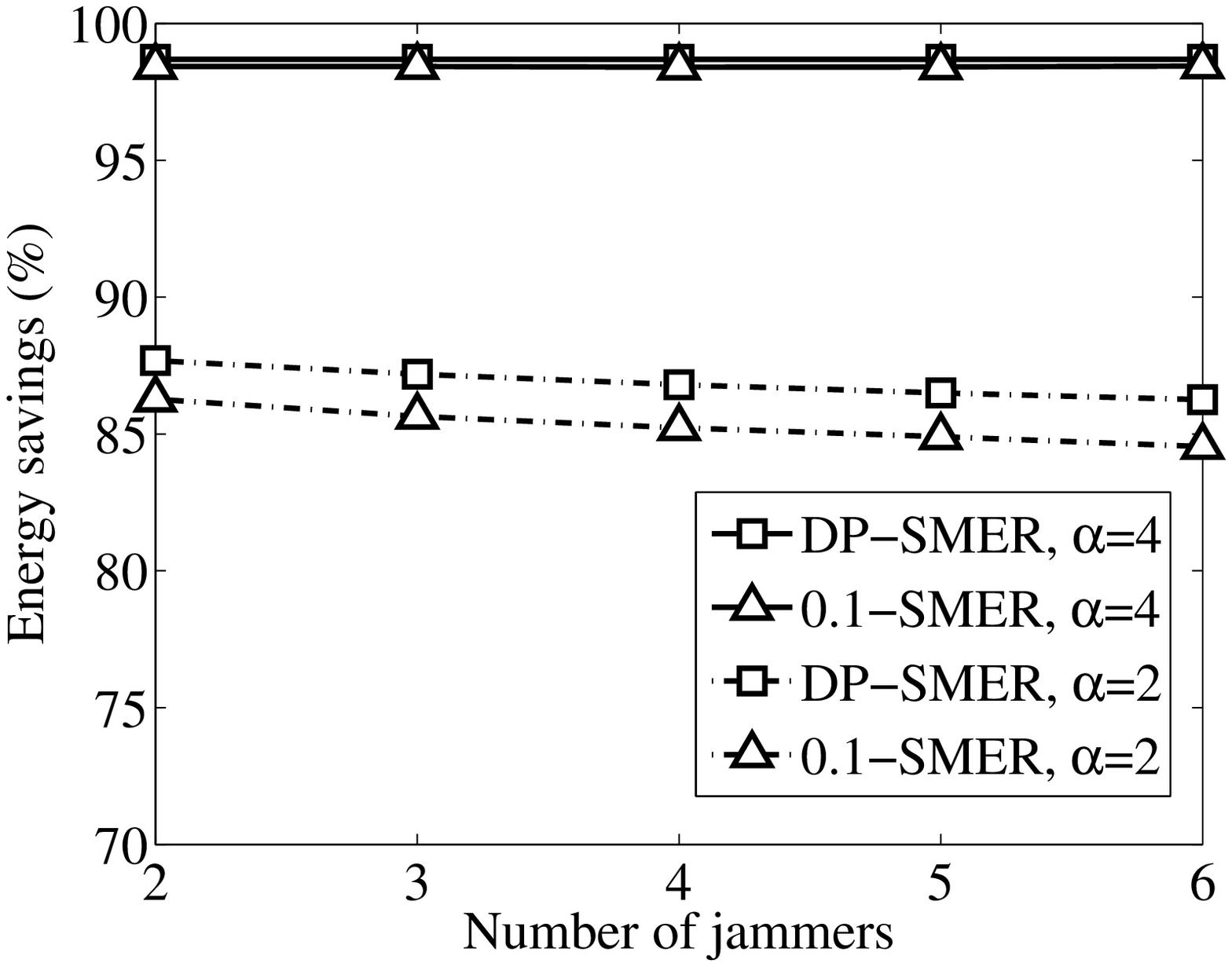}
        \label{fig:jammers_uniform}
    }
  \caption{Effect of the jamming set on energy savings.}
  \label{fig:jammers}
\end{figure}

\section{Related Work}
\label{sec:related}

A survey of prior work is presented in this section.

\spar{Secure Routing in Multi-hop Networks} While there are numerous
works on secure routing in wireless networks (see,
\eg~\cite{guizani08} and references therein), their focus is on
preventing malicious attacks that disrupt the operation of the
routing protocol using application level mechanisms such as
authentication and cryptography. The focus of this paper, on the
other hand, is on secure transmission of messages via the most
cost-effective paths in the network, which is orthogonal to the
secure routing problem considered in the existing literature.

\spar{Wireless Physical Layer Security}
The idea behind physical layer security is to exploit the
characteristics of the wireless channel such as fading to provide
secure wireless communications.
The foundations of information
theoretic security, which is the theoretical basis for physical layer
security, were laid by Wyner and
others~\cite{wyner1975wire,leung1978gaussian,csiszar1978broadcast}
based on Shannon's notion of perfect
secrecy~\cite{shannon1949communication}.
In the classical wiretap
model of Wyner, to achieve a strictly positive secrecy rate, the
legitimate user should have some advantage over the eavesdropper in
terms of SNR. Later, Maurer~\cite{maurer93} proved that even when a
legitimate user has a worse channel than an eavesdropper, it is
possible to have secure communication.
While some physical layer
security techniques allow for opportunistic exploitation of the
space/time/user diversity for secret
communications~\cite{maurer93,bloch2008}, others actively manipulate
the wireless channel to block eavesdroppers by employing techniques
such as multiple antennas~\cite{yates07} and
jamming~\cite{tekin2008general,elgamal08}. While some of these
techniques have been successfully implemented in practical
systems~\cite{katabi11}, physical layer security is focused on very
special network topologies, \eg\ single-hop networks. In this work,
we have developed algorithms to extend these techniques to multi-hop
networks.

\spar{Scaling Laws in Large Secure Networks} Motivated
by~\cite{gupta00}, recently, throughput scaling versus security
tradeoffs have been investigated in the context of large wireless
networks~\cite{liang2009secrecy, koyluoglu2010secrecy,
vasudevan2010security, Dennis2011JSAC}. Specifically, for cooperative
jamming when the eavesdroppers are uniformly distributed, it was
shown that if the number of eavesdroppers grows sub-linearly with
respect to the number of legitimate nodes, a positive throughput for
secure communication is achievable~\cite{Dennis2011JSAC}.

\spar{Security Based on Network Topology}
When there is sufficient path diversity in a network, different
messages can be routed over different parts of the network in the
hope that an eavesdropper would be incapable of capturing all
messages from across the network. To exploit network diversity for
security, various techniques based on multi-path
routing~\cite{fang09,krunz10} and network
coding~\cite{yeung02,jain04} have been investigated. While such techniques are suitable for wired networks,
their application in wireless networks is challenging due to lack of
path diversity at the source or destination of a communication
session. Moreover, there are considerable complications when splitting a flow among several paths, in particular,
at the granularity of a single session.
Moreover, network topology, in wireless networks, is a
function of power allocation at the physical-layer and propagation
environment, \eg\ fading.
Nevertheless, our approach is
complimentary to these techniques, by providing a mechanism to find a
minimum cost path that is information-theoretically secure,
regardless of the network diversity.

\vspace{-0.1cm}
\section{Conclusion}
\label{sec:conclusion}

This paper studied the problem of secure minimum energy routing in
wireless networks. It was shown that while the problem is
\mbox{NP-hard}, it admits exact pseudo-polynomial and fully polynomial time
$\epsilon$-approximation algorithmic solutions.
Furthermore, using simulations, we showed that our algorithms
significantly outperform security-agnostic algorithms based on
minimum energy routing. Finally, we note that our work can be
potentially extended to incorporate other secrecy models. Such extensions are left for future work.


\end{document}